\documentclass[a4paper,12pt]{article}

\usepackage{amsthm}
\usepackage{amsmath}
\usepackage{amssymb}
\usepackage{latexsym}
\usepackage{mathtools}
\usepackage{float}
\usepackage{url}
\usepackage{diagbox} 
\usepackage{threeparttable}
\usepackage[textwidth=18cm,textheight=20cm]{geometry}
\usepackage{cite} 
\usepackage[square,sort,comma,numbers]{natbib}
\usepackage{multirow}

\newtheorem{theorem}{Theorem}[section]
\newtheorem{proposition}[theorem]{Proposition}
\newtheorem{lemma}[theorem]{Lemma}
\newtheorem{corollary}[theorem]{Corollary}
\newtheorem{remark}[theorem]{Remark}
\newtheorem{definition}[theorem]{Definition}
\newtheorem{example}[theorem]{Example}

\newcommand{\LL}{\mathcal{L}}
\newcommand{\R}{\mathbb{R}}
\newcommand{\A}{\mathcal{A}}
\newcommand{\B}{\mathcal{B}}
\newcommand{\D}{\bar{D}}

\newcommand{\T}{\Theta}

\newcommand{\bb}{\bar{b}}
\newcommand{\C}{\mathcal{C}}
\newcommand{\s}{\hat{s}}
\newcommand{\es}{\mathcal{S}}
\newcommand{\tha}{\hat{\theta}}
\newcommand{\1}{\textbf{1}}

\usepackage{graphicx}

\begin{document}
\title{Two-fund separation under hyperbolically distributed returns and concave utility functions}

\author{\normalsize{$\text{Nuerxiati Abudurexiti}^a$} \quad \normalsize{$\text{Erhan Bayraktar}^b$} 
\quad \normalsize{$\text{Takaki Hayashi}^c$}
\quad
\normalsize{$\text{Hasanjan Sayit}^d$} \\
\\
\footnotesize{$\text{University of Michigan, USA}^b$}\\
\footnotesize{$\text{Keio University, Japan}^c$}\\
\footnotesize{$\text{Department of Management Science and Engineering, Xinjiang University, Urumqi, China}^a$}\\
\footnotesize{$\text{Xi’an Jiaotong-Liverpool University, Suzhou, China}^{a, d}$}
}

\date{\today}

\maketitle

\abstract{

Portfolio selection problems that optimize expected utility are usually difficult to solve. If the number of assets in the portfolio is large, such expected utility maximization problems become even harder to solve numerically. Therefore, analytical expressions for optimal portfolios are always preferred. In our work, we study portfolio optimization problems under the expected utility criterion for a wide range of utility functions, assuming return vectors follow hyperbolic distributions. 
Our main result demonstrates that under this setup, the two-fund monetary separation holds. Specifically, an individual with any utility function from this broad class will always choose to hold the same portfolio of risky assets, only adjusting the mix between this portfolio and a riskless asset based on their initial wealth and the specific utility function used for decision making. We provide explicit expressions for this mutual fund of risky assets. As a result, in our economic model, an individual's optimal portfolio is expressed in closed form as a linear combination of the riskless asset and the mutual fund of risky assets. Additionally, we discuss expected utility maximization problems under exponential utility functions over any domain of the portfolio set. In this part of our work, we show that the optimal portfolio in any given convex domain of the portfolio set either lies on the boundary of the domain or is the unique globally optimal portfolio within the entire domain.

\vspace{0.1in}
\vspace{0.1in}

\vspace{0.1in}

\section{Introduction}

Optimal asset allocation decisions are crucial for investors. These decisions involve choosing the optimal portfolio under a specific criterion, a problem that holds significant theoretical importance. The earliest work on portfolio choice is attributed to Markowitz. Following his pioneering mean-variance portfolio theory, numerous studies have proposed various criteria for portfolio selection.

One of the mainstream criteria for portfolio selection is the maximization of the expected utility function: a portfolio that provides the maximum expected utility of terminal wealth is considered optimal. In fact, the Markowitz mean-variance portfolio selection criterion is a special case of the expected utility criterion with a quadratic utility function. The tangent portfolios (or market portfolios) within the Markowitz mean-variance framework have closed-form analytical expressions because portfolio optimization problems with quadratic utility functions become quadratic optimization problems, for which closed-form solutions can be found.

If the utility function is not quadratic, there are only a few specific cases where the associated expected utility maximization problem leads to analytical solutions (see recent papers by \cite{Bodnar-Parola-Schmid-b}, \cite{Campbell-Viceira}, \cite{Canakoglu-Ozekici-2010}). For most other utility functions, numerical procedures are necessary (see \cite{Brandt-Santa}, \cite{Pirvu_Kwak}, and the references therein). However, these numerical procedures can be quite time-consuming if the portfolio space contains infinitely many elements. Therefore, analytical expressions for utility-maximizing optimal portfolios are always the preferred option. To achieve analytical expressions for expected utility-maximizing optimal portfolios, one must make specific distributional assumptions about the asset returns and the utility functions.

One of the most used distributions due to its simplicity is the normal distribution. But asset returns show leptokurtic features that normal distributions do not possess. Therefore modeling asset returns by proper distributions spawned an enormous amount of research in the past since the seminal work \cite{Mandelbrot_1963}. A proper model needs to capture most stylized facts observed in financial data. One of the important stylized facts that has been observed in equity prices is that a large price decline for  one stock is accompanied by a simultaneous large price drops for other stocks. This property of equity prices cannot be modeled properly by  multi-variate normal distributions. Empirical evidence supports heavy tails in financial data and this feature of multi-asset returns can be captured by heavy tailed Elliptical distributions with tail dependence. The class of Elliptical distributions allow for heavy tail models while preserving many of the attractive properties of the multi-variate normal models. It includes $t-$distributions, symmetric generalized hyperbolic distributions, sub-Gaussian $\alpha-$stable distributions etc. While the class of  Elliptical distributions is a broad class that offer flexible models, Elliptical distributions have their own limitations as they are radial symmetric. Empirical financial data (as reported in \cite{MCNeil-Frey-Embrechts-2015}) shows that lower tail dependence is often stronger than the upper tail dependence and this property cannot be captured by Elliptical distributions because of their radial symmetry. NMVM models includes elliptical distributions as a subclass and they offer flexible distributions in modeling radial asymmetry of financial data, while they keep the basic properties of Elliptical distributions. For example, the mixture variable $Z$, the location parameter $\mu$, the dispersion matrix $\Sigma$ remain the same with elliptical distributions and the only difference is in the skewness parameter $\gamma$ in our model (\ref{one}) (i.e., the case $\gamma=0$ corresponds to elliptical models).

As mentioned earlier, in our paper, we restrict return vectors to the NMVM models.
Based on this model assumption, we derive analytical expressions for expected utility-maximizing portfolios under a broad class of concave utility functions. Our assumption on the return vectors is not overly restrictive, especially considering their remarkable ability to capture most of the stylized features of financial assets. The NMVM models encompass many popular distributions used in financial modeling. For instance, they include the class of elliptical distributions as described by \cite{Owen-Rabinovitch-1983}, the multivariate $t$-distributions by \cite{Adcock-2010}, the multivariate variance gamma distributions by \cite{Luciano-Schoutens-2006}, and the multivariate GH distributions by \cite{prause1999generalized} and  \cite{MCNeil-Frey-Embrechts-2015}.  One specific example within NMVM that deserves mention is the generalized hyperbolic skewed $t$-distribution. As noted by \cite{Aas_Kjersti_And_Haff_Ingrid_Hobaek_2006}, the GH skewed $t$-distribution has one heavy tail and one semi-heavy tail, making it a suitable model for skewed and heavy-tailed financial data. The skewed $t$-distribution is an NMVM model where the mixing distribution $Z$ follows an inverse gamma distribution. 

Our main contribution in this paper is the two-fund separation result that is obtained in our main  Theorem \ref{3.155}. This result has important implications for portfolio managers. It simplifies the investment decision-making process by reducing the complex and computationally heavy portfolio selection process to the
selection between the skewness-induced tangent portfolio \footnote{The skewness-induced tangent portfolio $\mathcal{T}_{skew}$ is defined in Section 4}  and the risk-free asset. The skewness-induced tangent portfolio 
is universally optimal for all investors regardless of the investor's risk preference. Based on this portfolio,  investors can always create an optimal portfolio by merely adjusting the proportion of risk-free assets and the skewness-induced tangent portfolio. The traditional two-fund separation result (based on the mean-variance framework) has had significant impact on investment decisions and still remains key building block for portfolio management strategies today. Our result
extends this theory to the case of utility based portfolio selection framework 
reinforcing the importance of diversification in reducing unsystematic risk and overall portfolio risk. The results we obtained in this paper will help investors to make more informed decisions in the sophisticated real-world financial markets by utilizing the more realistic distributions NMVM (compared to Normal and elliptical distributions) in modeling multivariate asset returns.

In our paper we use the following notations. The notation $\R^d$ denotes the $d-$dimensional Euclidean space. We use $|x|=(x_1^2+x_2^2+\cdots x_d^2)^{1/2}$ to denote the Euclidean norm of a vector $x\in \R^d$. Vectors $\mu=(\mu_1, \mu_2, \cdots, \mu_d)^T$ and  $\gamma=(\gamma_1, \gamma_2, \cdots, \gamma_d)^T$ are elements of $\R^d$, where the  superscript $^T$ stands for the transpose of a vector or matrix. The notations $x\cdot \mu=x^T\mu=\sum_{i=1}^dx_i\mu_i$ denote the scalar product of the vectors $x$ and $\mu$. For any real valued $d\times d-$matrix $A$ we denote by $\Sigma=AA^T$ the product matrix and we use the notation $A=\Sigma^{1/2}$ to express that $A$ and $\Sigma$ are related by $\Sigma=AA^T$. The notation $L^p$ denotes the space of random variables with finite $p$ moments for any positive integer $p$, i.e., random variables $G$ with $E|G|^p<+\infty$. The notation $\B$ denotes the class of Borel subsets of $\R$ and $\B_+$ denotes the class of Borel subsets of $\R_+=(0, +\infty)$.

Our model set up in this paper is as follows. We consider a financial market with $d+1$ assets and the first asset is a risk-free asset with interest rate $r_f$ and the remaining $d$ assets are risky assets with log returns modelled by a $d-$dimensional random vector $X$. In this note, we assume that $X$ follows a  NMVM distribution. An $\R^d$-valued random variable $X$ is said to have an NMVM distribution if
\begin{equation}\label{one}
X=\mu+\gamma Z+\sqrt{Z}AN_d,
\end{equation}
where $\mu, \gamma \in \R^d$ are column vectors of dimension $d$,  $A\in \R^{d\times d}$ is a $d\times d$ matrix of real numbers,  $Z$ is a non-negative valued random variable that is independent from the $d-$dimensional standard normal random variable $N_d$. One can also define NMVM random vectors $X$ through their probability distribution functions $F$ on $(\R^d, \mathcal{B}^d)$. Namely, $X$ has   an NMVM distribution if 
\begin{equation*} \label{f}
F(dx)=\int_{\R_+}N_d(\mu+y\gamma, y\Sigma)(x)G(dy),    
\end{equation*}
where the mixing distribution $G$ is a probability measure on $(\R_+, \mathcal{B}_+)$ and $\Sigma=AA^T$. The short hand notation $F=N_d(\mu+y\gamma, y\Sigma)\circ G$ will be used quite often to denote NMVM return vectors in this paper.

Distributions  of the form (\ref{one}) show up quite often in continuous time financial modelling. For any  risky asset price $S_t\in \R^d_+$ log-returns over a time interval $\triangle$ are given by $X_i=\ln S_{t+\triangle}^{(i)}-\ln S_t^{(i)}\approx [S_{t+\triangle}^{(i)}-S_t^{(i)}]/S_t^{(i)}, i=1,2, \cdots, d$, where the approximation between log-returns and simple returns holds when the time interval $\triangle$ is small. If the risky asset prices are modelled as $S_t^{(i)}=S_0^{(i)}e^{X_t^{(i)}}, 1\le i\le d$, with $X_t^{(i)}, 1\le i\le d,$ being the components of the time changed Brownian motion model 
\begin{equation}\label{xtt}
X_t=\mu  t+\gamma \tau_t+B_{\tau_t},    
\end{equation}
where $B\in \R^d$ is a Brownian motion with zero mean and co-variance matrix $\Sigma=AA^T$ and $\tau_t$ is an independent subordinator (i.e., a non-negative L\'evy process with increasing sample paths), then the log-return vector of the price process $S_t$ has the distribution as in (\ref{one}). 

In fact, any model of the form (\ref{one}) induces a L\'evy process of the form (\ref{xtt}) that can be used in modelling log-return vector of risky asset prices as long as  the mixing distribution $Z$ is infinitely divisible, see Lemma 2.6 of \cite{Hammerstein_EAv_2010} for this. Exponential L\'evy price processes of this kind are quite popular in modelling risky asset prices, see \cite{Eberlein-Keller}, \cite{Barndorff-Nielsen_Ole_E_1997}, \cite{prause1999generalized}, \cite{Rydberg}. The model (\ref{xtt}) for  a log-return vector for risky asset prices are constructed by subordinating a Brownian motion with or without drift by subordinators. Similar procedure can be applied to construct log return processes of risky asset  prices  by using increasing and additive (independent and possibly non-homogeneous increments) processes, see Section 3.4 of \cite{P-H-D-MarkYor-self-decomposability} for such models. All these models have marginals distributed as (\ref{one}).

Given an initial endowment $W_0>0$, the investor needs to determine portfolio weights $x$  on the risky
assets such that the expected utility of the next period wealth is maximized. The wealth that corresponds to portfolio weight $x$ on the risky assets is given by
\begin{equation}\label{wealth}
\begin{split}
W(x)=&W_0[1+(1-x^T1)r_f+x^TX] \\
=&W_0(1+r_f)+W_0[x^T(X-\1 r_f)]
\end{split}
\end{equation}
 and the investor's problem is
\begin{equation}\label{L2}
\max_{x\in D}\; EU(W(x))
\end{equation}
for some domain $D$ of  the portfolio set $\R^d$. The main goal of this paper is to discuss the solutions of the problem (\ref{L2}) for a large class of utility functions $U$ when the risky assets have the NMVM distribution (\ref{one}). These type of utility maximization problems in one period models were studied in many papers in the past, see \cite{Madan_Mcphail}, \cite{Madan_Yen}, \cite{Pirvu_Kwak}, \cite{Zakamouline_Koekabakkar}, \cite{Birge_Chavez_2016}, \cite{Miklos-andrea}, \cite{Miklos-Lukasz}, \cite{Pirvu-Schulze}. Also see \cite{Deng-Pirvu} and the references there for multi-period utility maximization problems.

This paper is organized as follows. In Section 2,  we discuss the problem (\ref{L2}) under exponential utility function for any domain $D$ and show that when $D$ is a closed and convex domain the solution of (\ref{L2}) is either a global optimal portfolio on the entire portfolio domain $\R^d$ or it lies on the boundary $\partial D$ of $D$. We then  use this fact to give characterizations of optimal portfolios under short-sales constraint.  In Section 3, we first discuss the well-posedness of the problem (\ref{L2}) and introduce sufficient  conditions on the utility function $U$ that guarantee the existence of a solution for (\ref{L2}).  As one of the main results of the paper, we show that when the utility function is continuous, concave, and bounded from above,  the  expected utility maximization problem can be reduced to  a quadratic optimization problem. Using this result, we demonstrate that the solution of the portfolio optimization problem  (\ref{L2}) can be reduced to finding the maximum point of a real-valued function on the positive real-line.

\section{Portfolio optimization with exponential utility on convex domains}
In this section we study the optimization problem (\ref{L2})
for various domains $D$ of the portfolio set when the utility function is exponential. As stated earlier, when $D=R^d$ in (\ref{L2}), the recent paper \cite{Rasonyi-Sayit} showed that the corresponding optimal portfolio is unique and it is given by (\ref{themain}) below. We call the optimal portfolio (\ref{themain}) {\it globally optimal portfolio} for convenience.  The purpose of this section is to give some characterizations of the solutions for the problem (\ref{L2}) for various convex domains $D$ of the portfolio set. We then  give some characterizations of the optimal portfolios for the problem (\ref{L2}) when the domain $D$ is the set of portfolios with  short-sales constraints.

\subsection{Globally optimal portfolio}

Exponential utility maximization problems with models (\ref{one}) are studied in the paper  \cite{Rasonyi-Sayit} in detail. Since the results in \cite{Rasonyi-Sayit} are relevant to our work in this Section we first recall some definitions and notations from this paper. \begin{definition}\label{defs0} For any mixing distribution $Z$, if $\mathcal{L}_Z(s)<\infty$ for all $s\in \R$ we set $\s=-\infty$ and if $\mathcal{L}_Z(s)<\infty$ for some $s\in \R$ and $\mathcal{L}_Z(s)=+\infty$ for some $s\in \R$, we let $\s$ be the real number such that 
\begin{equation*}\label{s00}
\mathcal{L}_Z(s)=Ee^{-sZ}<\infty, \; \forall s>\s\;\;  \mbox{and} \;\; \mathcal{L}_Z(s)=Ee^{-sZ}=+\infty, \; \forall s<\s. 
\end{equation*}
We call $\s$ the critical value (we use the acronym CV-L from now on) of $Z$ in this paper. Observe that since $Z$ is non-negative random variable we always have $\s\le 0$.
\end{definition}

For a given  model (\ref{one}) introduce the following numbers  
\begin{equation}\label{ABC}
\mathcal{A}=\gamma^T\Sigma^{-1}\gamma,\; \mathcal{C}=(\mu-\1 r_f)^T\Sigma^{-1}(\mu-\1 r_f),\; \mathcal{B}= \gamma^T \Sigma^{-1}(\mu-\1 r_f).   
\end{equation}
Let $\hat{\theta}=\sqrt{\frac{\mathcal{A}-2\hat{s}}{\mathcal{C}}}$ and  define $\Theta$ as follows
\begin{equation}\label{Theta}
\Theta=\left \{
\begin{array}{ll}
(-\tha, \tha)&  \mbox{if $\s$ is finite and $\mathcal{L}_Z(\s)=+\infty$ or if $\s=-\infty$},\\
\left[ -\tha, \tha \right] &  \mbox{if $\s$ is finite and $\mathcal{L}_Z(\s)<\infty$}.\\
\end{array}
\right.
\end{equation}
Also define the function $Q(\theta)$ as
\begin{equation*}\label{H}
Q(\theta)=e^{\mathcal{C}\theta}\mathcal{L}_Z\Big [\frac{1}{2}\mathcal{A}-\frac{\theta^2}{2}\mathcal{C} \Big ].    
\end{equation*}
In its Lemma 4.1, the paper \cite{Rasonyi-Sayit} shows that the function $Q(\theta)$ is a strictly convex function. This fact is quite helpful for our discussions in the following sub-section. When the utility function is exponential $U(w)=-e^{aw}, a>0$, the paper \cite{Rasonyi-Sayit} shows that  the problem (\ref{L2}) with $D=\R^d$ has a closed form solution given by
\begin{equation}\label{themain}
x^{\star}=\frac{1}{aW_0}\Big [\Sigma^{-1}\gamma -q_{min}\Sigma^{-1}(\mu-\1 r_f)\Big ],    
\end{equation}
with
\begin{equation}\label{32q}
 q_{min}\in \arg min_{\theta \in \Theta}Q(\theta).
 \end{equation}
We cal this solution globally optimal portfolio in this paper. We also recall the Lemma 2.1 of \cite{Rasonyi-Sayit} here. According to this Lemma for any portfolio $x\in \R^d$ such that $EU(W(x))$ is finite the following holds
\begin{equation*}\label{L3}
EU(W(x))=-e^{-aW_0(1+r_f)}e^{-aW_0x^T(\mu-\1 r_f)}\mathcal{L}_Z\Big(aW_0x^T\gamma - \frac{a^2W_0^2}{2}x^T\Sigma x\Big ).
\end{equation*}
 In our discussions in this section, we use the following similar notations as in \cite{Rasonyi-Sayit}:
\begin{equation*}\label{Gx}
\begin{split}
g(x)=:&aW_0x^T\gamma-\frac{a^2W_0^2}{2}x^T\Sigma x,\\
G(x)=:&e^{-aW_0x^T(\mu-\1 r_f)}\mathcal{L}_Z\Big (aW_0x^T\gamma -\frac{a^2W_0^2}{2}x^T\Sigma x\Big ),\\
=&e^{-aW_0x^T(\mu-\1 r_f)}\mathcal{L}_Z\Big (g(x)\Big ).
\end{split}
\end{equation*}
With these notations we have the following obvious relation
\begin{equation*}\label{G-EU}
EU(W(x))=-e^{-aW_0(1+r_f)}G(x).  
\end{equation*}
Therefore maximizing $EU(W(x))$ on some domain $D$ is equivalent to minimizing $G(x)$ in the same domain.

\subsection{Optimal portfolios in convex domains are either globally optimal or lie on the boundary}
The main goal of this section is to characterize the solution of (\ref{L2}) for any given convex and closed domain $D$. Throughout the paper, as in \cite{Rasonyi-Sayit}, we make the assumption $\mu-\1 r_f\neq 0$ so that the $\C$ in (\ref{ABC}) is not zero. Before we discuss the solutions of the problem (\ref{L2}), we first prove the following Lemma \ref{lem22.6}. 

\begin{lemma}\label{lem22.6} Consider the utility optimization problem (\ref{L2}) with $U(w)=-e^{-aw}, a>0$. Assume the domain $D$ is a closed and convex subset of $\R^d$ and assume (\ref{L2}) has a solution $x_0\in D$. Then it is unique and it solves
\begin{equation}\label{Qoptt}
\begin{split}
 \max_{x} \; \;  & aW_0x^T\gamma - \frac{a^2W_0^2}{2}x^T\Sigma x, \\
s.t.\;\;  & x^T(\mu-r_f\1)=c_0,\\
& x\in D,
\end{split}
\end{equation}
for some $c_0\in \R$. Define $\bar{D}=\{x^T(\mu-r_f\1): x\in D\}$ then we have $c_0\in \bar{D}$. 
\end{lemma}

\begin{proof} For $x_0$ we define $c_0=:x_0^T(\mu-\1 r_f)$ first. Now, let $x_1\in D$ be the solution to the problem (\ref{Qoptt}) (such an $x_1$ exists as $D$ is a closed and convex set and at the same time portfolios with large Euclidean norm drives the objective function in (\ref{Qoptt}) to negative infinity). We need to show $x_0=x_1$. The solution $x_1$ is unique as $\Sigma$ is positive definite by the assumption of the model (\ref{one}) and $D$ is  a convex set. Then by the optimality of $x_1$, we have $g(x_0)\le g(x_1)$. Since $\mathcal{L}_Z(s)$ is a decreasing function we have $\mathcal{L}_Z(g(x_1))\le \mathcal{L}_Z(g(x_0))$. Since $c_0=x^T_0(\mu-\1 r_f)=x_1^T(\mu-\1 r_f)$
we have $G(x_1)\le G(x_0)$. This shows that $EU(W(x_1))\geq EU(W(x_0))$. But $x_0$ is optimal for (\ref{L2}). Therefore we should have $EU(W(x_0))=EU(W(x_1))$. This implies $G(x_0)=G(x_1)$ and this in turn implies that
$g(x_0)=g(x_1)$ again due to $c_0=x^T_0(\mu-\1 r_f)=x_1^T(\mu-\1 r_f)$. The uniqueness of the optimization point for (\ref{Qoptt}) then implies
$x_0=x_1$. Now since $x_0\in D$ we have $c_0\in \bar{D}$. This completes the proof.
\end{proof}
\begin{remark} In the above Lemma \ref{lem22.6}, it is not necessary to require, in addition to convexity and closedness,  the boundedness of the domain $D$. This is because for any portfolio sequence with divergent Euclidean norm the expected utility in (\ref{L2}) diverges to negative infinity when $U$ is an exponential utility function as will be illustrated in sub-section 3.1 below. Hence both $x_0$ and $x_1$ in the proof of this Lemma are portfolios with finite Euclidean norms.

\end{remark}
\begin{remark}\label{rem2.2} The above Lemma \ref{lem22.6} gives a characterization of the optimal portfolio for the problem (\ref{L2}). According to this Lemma the solution for (\ref{L2}) can be obtained by solving a constrained quadratic optimization problem in (\ref{Qoptt}). This clearly simplifies the calculation of the optimal portfolio for the problem (\ref{L2}) as the quadratic optimization problem (\ref{Qoptt}) is a simpler problem to solve. However, the constant $c_0$ in (\ref{Qoptt}) is not given explicitly and hence 
one needs to solve the optimization problem (\ref{Qoptt}) for each $c_0\in \bar{D}$. This can be quite time consuming sometimes. In this section we attempt to provide further characterizations of the solution to the problem (\ref{L2}) when $D$ is a convex and closed  subset of $\R^d$.
\end{remark}

The above Lemma \ref{lem22.6} characterizes the solution of (\ref{L2}) for any closed and convex domain. It shows in particular that if a solution for (\ref{L2}) exists then it solves a quadratic optimization problem. As pointed out in Remark \ref{rem2.2} above, finding the solution for (\ref{L2}) by using this Lemma is time consuming. We wish to give some simpler and less time consuming approach for the solution of the problem (\ref{L2}). To this end, first we consider the following optimization problem 
\begin{equation}\label{Qopt-noS}
\begin{split}
 \max_{x} \; \;  & aW_0x^T\gamma - \frac{a^2W_0^2}{2}x^T\Sigma x, \\
s.t.\;\;  & x^T(\mu-r_f\1)=c,\\
\end{split}
\end{equation}
for any given $c\in \R$. Lemma 2.12 of \cite{Rasonyi-Sayit} gives the optimal solution of (\ref{Qopt-noS}) as
 \begin{equation*}\label{xc0}
x_{c}=\frac{1}{aW_0}\Big [\Sigma^{-1}\gamma-q_{c}\Sigma^{-1}(\mu-\1 r_f)\Big ],   
 \end{equation*}
where
\begin{equation}\label{qcc}
q_{c}=\frac{\gamma^T\Sigma^{-1}(\mu-\1 r_f)-aW_0c}{(\mu-\1 r_f)^T\Sigma^{-1}(\mu-\1 r_f)}=\frac{\B-aW_0c}{\C}.    
\end{equation}
The same Lemma in \cite{Rasonyi-Sayit} shows that
\begin{equation*}\label{gq}
g(x_c)=\frac{\A}{2}-\frac{\C}{2}q_c^2.    
\end{equation*}

This fact will be of great help to our analysis in this section. For each portfolio domain $D$ we define the following sets
\begin{equation}\label{Dq}
 \bar{D}=\{x^T(\mu-r_f\1): x\in D\}, \; \; D_q=\{q_c: c\in \bar{D}\}, \; \; \bar{D}_q=D_q\cap \Theta,   
\end{equation}
where $q_c$ is defined as in (\ref{qcc}). The set $\bar{D}_q$ depends on $\Theta$ through (\ref{Theta}). Thus it depends on the CV-L of the mixing distribution $Z$ in (\ref{one}).

Recall that our objective is to obtain closed form solution for (\ref{L2}) for a given convex and closed domain $D$. This problem is a complex problem as the domain $D$ can be any. The optimal portfolio that solves (\ref{L2}), if it exists, can be on the interior $int(D)$ of the domain $D$ (here ``int'' denotes the interior in the Euclidean norm in $\R^d$) or it can be on the boundary $\partial D$ of it. But if the optimal portfolio for (\ref{L2}) lies in $int(D)$ then we have a closed form expression to it as the following Lemma \ref{prop2.3} shows. Before we state this Lemma and give its proof we first write down a remark.
\begin{remark}\label{rem-Q} Consider a general quadratic optimization problem 
\begin{equation}\label{27}
 \begin{split}
  \text{maximize} \quad & m^Tx-\frac{1}{2}x^THx\\
  \text{subject to} \quad & \Gamma x=p, 
 \end{split}   
\end{equation}
where $H\in \R^d\times \R^d$ is a symmetric positive definite matrix, $m\in \R^d$, $\Gamma \in \R^{k\times d}$, and $p\in \R^k$. This is a quadratic optimization problem on the affine set $\Gamma x=p$. Assume $\Gamma x=p$ has a solution $\hat{x}$. Then the set of all solutions of $\Gamma x=p$ is given by $\hat{x}+O$, where $O=Null(\Gamma)$. So the problem (\ref{27}) is equal to the problem $max_{x\in \hat{x}+O}[m^Tx-\frac{1}{2}x^THx]$. 

It is well known that a local solution to the problem (\ref{27}) exists if and only if $H$ is positive semi-definite and there exists a vector $\bar{x}\in \hat{x}+O$ such that $H\bar{x}+m\in O^{\perp}$ (the orthogonal space of $O$), in which case $\bar{x}$ is a local solution for (\ref{27}). It is also well known that if $\bar{x}$ is a local solution for (\ref{27}), then it is a global solution, a fact that is useful for the proof of Lemma \ref{prop2.3} below. The problem (\ref{27}) has a unique global solution if and only if $H$ is positive definite. 
\end{remark}

\begin{lemma}\label{prop2.3} Consider the utility optimization problem (\ref{L2}) with $U(w)=-e^{-aw}, a>0$. Assume $\A\neq 0$ or $\s\neq 0$. Assume the domain $D$ is a closed and convex subset of $\R^d$. Let $x_0\in D$ be a solution for (\ref{L2}) and assume $x_0\in int(D)$. Then 
\begin{equation}\label{x00}
 x_0=\frac{1}{aW_0}\Big [\Sigma^{-1}\gamma -q_d\Sigma^{-1}(\mu-\1 r_f)\Big ],     
\end{equation}
 with
\begin{equation*}
q_d=\arg min_{\theta \in \bar{D}_q}Q(\theta),
\end{equation*}
where $\bar{D}_q$ is defined as in (\ref{Dq}).
\end{lemma}
\begin{proof} By Lemma \ref{lem22.6}, $x_0$ solves the quadratic optimization problem (\ref{Qoptt}) for some $c_0\in \bar{D}$. Let $D_0=D\cap \{x: x^T(\mu-r_f\1)=c_0\}$. Then $D_0$ is, being the intersection of two convex sets, a convex domain on the hyperplane $x^T(\mu-r_f\1)=c_0$. Also $x_0$ belongs to $rel-int(D_0)$ (here ``rel-int'' denotes the interior in the relative topology of the hyperplane $x^T(\mu-r_f\1)=c_0$). Therefore $x_0$ is a local solution of the quadratic function in (\ref{Qoptt}) on the affine set $x^T(\mu-r_f\1)=c_0$. From this we conclude that  
$x_0$ is a global solution of the quadratic function in (\ref{Qoptt}) on the affine set $x^T(\mu-r_f\1)=c_0$ as explained in the
Remark \ref{rem2.2} above. Then Lemma 2.12 of \cite{Rasonyi-Sayit} shows that
\[
x_0=\frac{1}{aW_0}[\Sigma^{-1}\gamma-q_{c_0}\Sigma^{-1}(\mu-\1 r_f)],
\]
where $q_{c_0}=\B/\C-[aW_0/\C] c_0$. Also from Lemma 2.13 of \cite{Rasonyi-Sayit} we have 
$G(x_0)=e^{-\B}Q(q_{c_0})$. Observe that since at $x=x_0$ the expected utility is finite (by the assumptions in the Lemma), we have $q_{c_0}\in \T$. Now, define the following map 
\[
x(q)=:\frac{1}{aW_0}[\Sigma^{-1}\gamma-q\Sigma^{-1}(\mu-\1 r_f)]
\]
of $q$. We can easily calculate the following
\[
g(x(q))=\frac{\A}{2}-\frac{\C}{2}q^2, \; \; G(x(q))=e^{-\B+q\C}\LL_Z(\frac{\A}{2}-\frac{\C}{2}q^2)=e^{-\B}Q(q).
\]

Now, since $x_0=x(q_{c_0})$ is an interior point of $D$, there exists a $\delta_0>0$ such that for any $q\in (q_{c_0}-\delta_0, q_{c_0}+\delta_0)$ we have $x(q)\in D$. Then the optimality of $x_0$ in the domain $D$ implies that  $G(x_0)\le G(x(q))$ for all $q\in (q_{c_0}-\delta_0, q_{c_0}+\delta_0)$ (here we don't rule out the possibility that $G(x(q))=+\infty$ for all $q\neq q_{c_0}$). This in turn implies that $Q(q_{c_0})\le Q(q)$ for all $q\in (q_{c_0}-\delta_0, q_{c_0}+\delta_0)$. Observe here that we have $q_{c_0}\in D_q$ also and therefore $q_{c_0}\in \D_q$ (as $q_{c_0}\in \T$ also). Now since 
$Q(\theta)$ is a strictly convex function on $\T$ (as shown in the Lemma 4.1 of \cite{Rasonyi-Sayit}) it is strictly convex on $\D_q$ also. Therefore any local minimum of $Q(\theta)$ on $\bar{D}_q$ is in fact a global minimum on  $\bar{D}_q$. From these we can conclude that $q_{c_0}$ should be the minimizing point of $Q(\theta)$ in $\D_q$. This complete the proof.
\end{proof}
\begin{remark} In the Lemma \ref{prop2.3}
above, the condition $x_0\in int(D)$ is stated. Here it is important to note that $int(D)$ is the interior under the Euclidean norm topology of $\R^d$. For any convex domain on hyperplanes in $\R^d$, the Lemma \ref{lem22.6} holds also. However, the interior of such convex domains  in the Euclidean  topology of $\R^d$ are empty sets. So in these cases our above Lemma \ref{prop2.3} is not applicable.
\end{remark}

\begin{proposition}\label{cor2.55} Consider the utility optimization problem (\ref{L2}) with $U(w)=-e^{-aw}, a>0$. Assume $\A\neq 0$ or $\s\neq 0$. Assume the domain $D$ is a closed and convex subset of $\R^d$. Let $x_0\in D$ be a solution for (\ref{L2}). Then either $x_0$ is the unique  solution for 
\begin{equation}\label{RL2}
\max_{x\in \R^d}EU(W(x))
\end{equation}
and it is given by (\ref{themain})
or $x_0\in \partial D$.    
\end{proposition}
\begin{proof} Assume $x_0\notin \partial D$, then $x_0\in int(D)$. From above Lemma \ref{prop2.3}, $x_0$ has the expression as in (\ref{x00}) with the corresponding $q_d\in \Theta$. Since $x_0$ is the interior point of $D$, there exists $\delta_0>0$ such that for all $q\in (q_d-\delta_0, q_d+\delta_0)$ the portfolios   
\[
x(q)=\frac{1}{aW_0}[\Sigma^{-1}\gamma-q\Sigma^{-1}(\mu-\1 r_f)]
\]
are in $D$. As explained in the proof of the above Lemma \ref{prop2.3}, we have $G(x(q))=e^{-\B}Q(q)$ for all $q\in (q_d-\delta_0, q_d+\delta_0)$. Observe that $x_0=x(q_d)$. The optimality of $x_0$ in $D$ implies 
$G(x(q_d))\le G(x_q)$ for all $q\in (q_d-\delta_0, q_d+\delta_0)$. From these we conclude $Q(q_d)\le Q(q)$ for all $q\in (q_d-\delta_0, q_d+\delta_0)$. But $Q(\theta)$ is a strictly convex function on $\Theta$ as shown in Lemma 4.1 of \cite{Rasonyi-Sayit}. Thus $q_{d}$ should be the unique minimizing point 
of $Q(\theta)$ in $\T$. Then from Theorem 2.15 of \cite{Rasonyi-Sayit}, we conclude that $x_0$ is the solution for (\ref{RL2}).
\end{proof}

\begin{remark}\label{rem2.9} The results of this section claim that if a solution $x_0$ to the problem (\ref{L2}) exists under the condition that $D$ is convex and closed, then it either lies in $int(D)$ and takes the form (\ref{x00}) or it lies on the boundary $\partial D$. Here we did not address the problem of the existence of a solution for (\ref{L2}) when $D$ is a convex and closed subset of $\R^d$. In Section 3 below we will show that the solution to (\ref{L2}) always exists as long as  $D$ is a closed subset of $\R^d$ with at least one point $x_0\in D$ with $-\infty<EU(W(x_0))<+\infty$. Therefore  for any given  convex and closed domain $D$ that satisfies this condition, if it does not contain the global  portfolio (\ref{themain}) in it, then optimal portfolio
for the problem (\ref{L2}) needs to be searched from the boundary $\partial D$ of $D$. As we shall see, this fact will be  helpful to obtain some closed form expressions for possible solutions to the problem (\ref{L2}) under short-sales  constraints. 
\end{remark}

\subsection{Optimal portfolios under short-sales  constraints}

Next we address the problem of finding optimal portfolio for the problem (\ref{L2}) under short-sales  constraints. Short sale is the sale of a stock that the seller does not own.  These transactions are settled by the delivery of borrowed stock. A short seller needs to close out the position later by returning the borrowed stock to the  lender. Due to transaction costs and  various  market frictions, short selling  is often  not a favorable option for many traders.  Additionally, government regulations  sometimes do not permit for short sales. Therefore it is useful to study the optimization problem (\ref{L2}) under short-sales constraints.

In this section we assume that short sales of the risky assets are not permitted. But long and short positions on the risk-free asset are allowed. Under this assumption, the portfolio space with short-sales constraints is given by
\begin{equation}\label{short-33}
 \mathcal{S}=\{x\in \R^d: x_i\geq 0\}.   
\end{equation}
The  portfolio optimization problem under short-sales constraints is then given by 
\begin{equation}\label{SL2}
\max_{x\in \mathcal{S}}\; EU(W(x)).
\end{equation}
As $\mathcal{S}$ is convex and closed, by Lemma \ref{lem22.6} the corresponding optimal portfolio (see Remark \ref{rem2.9}) solves the following constrained quadratic optimization problem
\begin{equation*}\label{Qopt-S}
\begin{split}
 \max_{x} \; \;  & aW_0x^T\gamma - \frac{a^2W_0^2}{2}x^T\Sigma x, \\
s.t.\;\;  & x^T(\mu-r_f\1)=c,\\
& x\in \mathcal{S},
\end{split}
\end{equation*}
for some $c\in \bar{\es}=\{x^T(\mu-r_f\1): x\in \es\}$. 

The corresponding sets in (\ref{Dq}) are given by 
\[
\bar{\es}=\{x^T(\mu-r_f\1): x\in \es\}, \es_q=\{q_c: c\in \bar{\es}\}, \bar{\es}_q=\es_q\cap \Theta.
\]
Assume the solution for (\ref{SL2}) is in the interior of $\es$ and assume $\A\neq 0$ or $\s\neq 0$ is satisfied, then  by Lemma \ref{prop2.3} it is given by \begin{equation}\label{34}
 x_{\es}=\frac{1}{aW_0}\Big [\Sigma^{-1}\gamma -q_{\es}\Sigma^{-1}(\mu-\1 r_f)\Big ],     
\end{equation}
where $q_{\es}=\arg min_{\theta \in \bar{\es}_q} Q(\theta)$. In this case, as explained in Proposition \ref{cor2.55} above, the solution (\ref{34}) is in fact the global solution  of (\ref{RL2}). 

If (\ref{34}) does not turn out to be the solution for (\ref{SL2}), then we need to look for the solution of (\ref{SL2}) from the boundary $\partial \es$. To describe $\partial \es$, we denote by $I$ any non-empty subset of $\bar{d}=:\{1, 2, \cdots, d\}$ and let $J=\bar{d}/I$. Define 
\[
\partial \es_J=:\{x\in \R^d: x_i=0, i\in I, x_j>0, j\in J\}
\]
Then $\partial \es=\cup_I\partial \es_I$, where the union is over all non-empty subset $I$ of $\bar{d}$. We introduce the following projection 
\begin{equation*}\label{pj}
P_J: \R^d\rightarrow \R^J, P_Jx=x_J,    
\end{equation*}
where $x_J$ is $J-$dimensional vector composed of the $j'th$ rows of $x$ for all $j\in J$ (not changing the order, for an example if $x^T=(4, 5, 6)\in \R^3$ and $J=\{1, 3\}$, then $x_J=(4, 6)^T$). The inverse map of $P_J$ is denoted by $P_J^{-1}$, i.e., $P_J^{-1}x_J=x$.

Now, if the solution $x_0$ for (\ref{SL2}) is on the boundary $\partial \es$, then $x_0\in \partial \es_I$ for some  non-empty $I\subset \bar{d}$. If $J=\emptyset$ (which means $I=\bar{d}$), then $x_0=0=\partial \mathcal{S}_J$ is the zero portfolio. So we assume $J$ is not empty below. With a given model (\ref{one}), denote $\mu_J=P_J\mu, \gamma_J=P_J\gamma,$ and let $A_J$ be $J\times J$ matrix obtained by deleting $i'th$ columns and $i'th$ rows of the matrix $A$ in (\ref{one}) for all $i\in I$. Define the random vector
\begin{equation*}
X_J=\mu_J+\gamma_JZ+\sqrt{Z}A_JN_J,    
\end{equation*}
where $N_J$ is the $J-$dimensional standard normal random variable. The wealth that corresponds to the return vector $X_J$ above is
\begin{equation*}
W_J(x_J)=W_0(1+r_f)+W_0[x_J^T(X_J-\1r_f)].
\end{equation*}

For any portfolio $x\in \R^d$ let 
 $x_J=P_Jx$, then we have 
\begin{equation*}
W(x)=W_J(x_J),    
\end{equation*}
as long as $x\in \partial \es_J$. Therefore optimizing $EU(W(x))$ on $\partial \es_I$ becomes a problem of optimizing $EU(W_J(x_J))$. But from \cite{Rasonyi-Sayit} we know the optimal portfolio for 
\begin{equation}\label{optJ}
 \max_{x_J}EU(W_J(x_J)).   
\end{equation}
To be able to write down the solution for (\ref{optJ}) by using the results in \cite{Rasonyi-Huy}, we define
\begin{equation*}\label{HJ}
Q_J(\theta)=e^{\mathcal{C}_J\theta}\mathcal{L}_Z\Big [\frac{1}{2}\mathcal{A}_J-\frac{\theta^2}{2}\mathcal{C}_J \Big ],    
\end{equation*}
where
\begin{equation*}\label{ABCJ}
\mathcal{A}_J=\gamma^T_J\Sigma^{-1}_J\gamma_J,\; \mathcal{C}=(\mu_J-\1 r_f)^T\Sigma^{-1}_J(\mu_J-\1 r_f),\; \mathcal{B}_J= \gamma^T_J \Sigma^{-1}(\mu_J-\1 r_f),   
\end{equation*}
and $\Sigma_J=A_J^TA_J$. Let
\begin{equation*}
 q_{min}^J\in \arg min_{\theta \in \Theta}Q_J(\theta),  \end{equation*}
and 
\begin{equation}\label{44}
x^{\star}_J=\frac{1}{aW_0}\Big [\Sigma^{-1}_J\gamma_J -q_{min}^J\Sigma^{-1}_J(\mu_J-\1 r_f)\Big ].    
\end{equation}

\begin{proposition}\label{short} The optimal portfolio for the problem (\ref{SL2}) is either the zero portfolio $x^{\star}=0$ or it is given by $P_J^{-1}x_J^{\star}$ for some non-empty $J\subset \bar{d}$, where $x_J^{\star}$ is given by (\ref{44}). 
\end{proposition}
\begin{remark} In the above Proposition \ref{short}, we did not list (\ref{themain}) as one of the possibilities for the optimal portfolio under short-sales constraints, as this case is covered by $X_J^{\star}$ with $J=\bar{d}$. Additionally, we should mention that the problem (\ref{short-33}) is well-defined  as will be discussed in Lemma \ref{opt-D} in Section 3 below.   
\end{remark}

Next we consider the following optimization problem:
\begin{equation}\label{EE}
\begin{split}
 \max_{x} \; \;  & EU(W(x)), \\
s.t.\;\;  & EW(x)\geq \ell,
\end{split}
\end{equation}
for some given level $\ell$. For the well-posedness of this problem we assume that there exists at least one portfolio $\bar{x}$ with $EW(\bar{x})\geq \ell$ in the rest of this section, see Lemma \ref{opt-D} in Section 3 for this.

Here we optimize the expected utility under the condition that the expected wealth stay above a given level $\ell$. From (\ref{wealth}), the expected wealth of a portfolio is given by 
\[
EW(x)=W_0(1+r_f)+W_0x^T(\mu-r_f\1+EZ\gamma).
\]
Then the constraint $EW(x)\geq \ell$  is equivalent to $x^T(\mu-r_f\1+EZ\gamma)\geq [\ell-W_0(1+r_f)]/W_0$. If we denote $v=:\mu-r_f\1+EZ\gamma$ and $r=\ell-W_0(1+r_f)]/W_0$, then the optimization problem (\ref{EE})  becomes 
\begin{equation}\label{EEv}
\begin{split}
 \max_{x} \; \;  & EU(W(x)), \\
s.t.\;\;  & x^Tv\geq r.
\end{split}
\end{equation}
From Lemma (\ref{lem22.6}), the solution for this problem  solves the following constrained quadratic optimization problem.

\begin{equation*}\label{Qopt-Sr}
\begin{split}
 \max_{x} \; \;  & aW_0x^T\gamma - \frac{a^2W_0^2}{2}x^T\Sigma x, \\
s.t.\;\;  & x^T(\mu-r_f\1)=c_0,\\
& x^Tv\geq r.
\end{split}
\end{equation*}

Based on these facts we can state the following Corollary.

\begin{corollary} Consider the optimization problem (\ref{EE}). If the optimal portfolio $x_0$ for (\ref{EEv}) satisfies $x_0^Tv>r$ (interior of the convex domain $x^Tv\geq r$), then by Lemma \ref{prop2.3} the solution takes the form (\ref{x00}). If the solution is not in the interior then $x_0$ satisfies
\begin{equation}\label{Qopt-Srr}
\begin{split}
 \max_{x} \; \;  & aW_0x^T\gamma - \frac{a^2W_0^2}{2}x^T\Sigma x, \\
s.t.\;\;  & x^T(\mu-r_f\1)=c_0,\\
& x^Tv=r.
\end{split}
\end{equation}
\end{corollary}
We remark here that it is not difficult to obtain closed form solution for this problem (\ref{Qopt-Srr}) as this is a quadratic optimization problem. Next we present some examples as an application of our results in this section.
\begin{example}\label{ex-short1} Consider a financial market with one risk-free asset and one risky asset. Assume the log return of the risk free asset in one period of time is $\ln\frac{B_{t+1}}{B_t}=r_f$ and the log return of the risky asset in one period is given by
\begin{equation}\label{one-p}
\ln \frac{S_{t+1}}{S_t}\overset{d}{=}b_1+b_2Z+b_3\sqrt{Z}N(0, 1), 
\end{equation}
 where $b_1, b_2 \in \R, b_3>0$, and $Z$ is a non-negative random variable independent from $N(0, 1)$.  Assume for simplicity that the CV-L of $Z$ is a finite number, i.e., $\hat{s}>-\infty$ and $\LL_Z(\s)<+\infty$. We also assume $b_1\neq r_f$.

Let $x$ denote the fraction of the initial wealth $W_0$ invested on the risky asset for an exponential utility maximizer with utility function $U(w)=-e^{-aw}, a>0$. The exponential utility maximizer is interested to find out her/his optimal investment under short-sales constraint. Namely she/he is interested to find the solution to the following problem
\begin{equation}\label{S-SL2}
\max_{x\geq 0}EU(W(x)).
\end{equation}
Due to (\ref{themain}),  the global solution to the problem $\max_{x\in \R}EU(W(x))$ without short-sales  constraints is given by
\[
x^{\star}=\frac{1}{aW_0b_3^2}[b_2-q_{min}(b_1-r_f)],
\]
where $q_{min}=argmin_{\theta\in \Theta}$ and 
\begin{equation*}
\begin{split}
 Q(\theta)=&e^{\frac{(b_2-r_f)^2}{b_3^2}\theta}\LL_Z(\frac{b_2^2}{2b_3^2}-\frac{(b_2-r_f)^2}{2b_3^2}\theta^2),\\ 
\hat{\theta}=&\sqrt{(b_2^2-2\hat{s}b_3^2)/(b_1-r_f)^2}, \\
 \Theta=&[-\hat{\theta},
 \hat{\theta}],\\
 \end{split}
\end{equation*}
If $b_2-q_{min}(b_1-r_f)>0$, then the above $x^{\star}$ is the solution for (\ref{S-SL2}). If $b_2-q_{min}(b_1-r_f)\le 0$, then the solution of (\ref{S-SL2}) needs to lie on the boundary of the domain $D=:\{x\geq 0\}$ by Proposition \ref{cor2.55} above. But $\partial D=\{0\}$. From these we conclude that the solution to the problem (\ref{S-SL2}) is either given by $x^{\star}$ above or it is the zero portfolio (which corresponds to investing everything on the risk-free asset).
\end{example}
We remark here that models with one period log returns of the form (\ref{one-p}) are quite popular in financial modelling. For example, if $Z$ is a gamma random variable then the stock price process $S_t$ corresponds to the exponential of the variance gamma process, see \cite{Madan_Dilip_B_And_Carr_Peter_P_And_Chang_Eric_C_1998}. If $Z$ is an inverse Gaussian random variable, then $S_t$ is the exponential of the  Normal inverse Gaussian L\'evy process, \cite{Madan_Dilip_B_And_Seneta_Eugene_1990}.

\begin{example} Now consider a financial market with three assets: one risk-free and two risky assets. The log-return of the risk-free asset is given by $\ln \frac{B_{t+1}}{B_t}=r_f$ as in the above Example \ref{ex-short1}. The vector of one period log returns $X=[\ln (S_{t+1}^{(1)}/S_t^{(1)}), \ln (S_{t+1}^{(2)}/S_t^{(2)})]$ of the two risky assets is given by (\ref{one}). Let $x_1$ denote the fraction of initial wealth $W_0$ that an exponential utility maximizer invests on asset $S_t^{(1)}$ at time $t$ and similarly let $x_2$ be the fraction of initial wealth invested on the risky asset $S_t^{(2)}$. An exponential utility maximizer is interested to solve the following optimization problem:  
\begin{equation}\label{SS-SL2}
\max_{x_1\geq 0, x_2\geq 0}EU(W(x)).
\end{equation}
One possibility of the optimal portfolio for the problem (\ref{SS-SL2}) is given by (\ref{themain}) which we denote  $x^{\star}(1)=:(x_1^{\star}, x_2^{\star})$. If $x_1^{\star}\geq 0$ and  $x_2^{\star}\geq 0$, then $x^{\star}(1)$ is the solution for (\ref{SS-SL2}). If one of $x_1^{\star}, x_2^{\star}$ is strictly less than zero, then by Proposition \ref{short}, the optimal portfolio for (\ref{SS-SL2}) should be searched from the boundary of the domain $D=\{x_1\geq 0, x_2\geq 0\}$. One possible optimal portfolio is the zero vector $x^{\star}(2)=:(0, 0)$. The other possibilities, again by Proposition \ref{short}, are given as follows: Let $x_1^{\star}$ be the optimizing portfolio of the same exponential utility maximizer for the case of  one-period log returns at time $t$ in the market $(B_t, S_t^{(1)})$ and similarly let $x_2^{\star}$ be the  optimizing portfolio for the case of one-period log returns at time $t$ in the market $(B_t, S_t^{(2)})$. Then either $x^{\star}(3)=:(x_1^{\star}, 0)$ or $x^{\star}(4)=:(0, x_2^{\star})$ are optimal portfolios for (\ref{SS-SL2}). In summary, in the three asset economy presented in this example one of $x^{\star}(1), x^{\star}(2), x^{\star}(3), x^{\star}(4)$ is a solution for (\ref{SS-SL2}). Here only one of these four possibilities is a solution for (\ref{SS-SL2}) as the solution for (\ref{L2}) when $D$ is a convex domain is unique. 
\end{example}

\section{Portfolio optimization with general utility functions}

As pointed out in the introduction there are very narrow cases of utility functions that the associated expected utility maximization problem lead into  
analytical solution for optimal portfolio. In this Section we provide analytical solution to the following problem
\begin{equation}\label{new-L2}
\max_{x\in \R^d}\; EU(W(x)),
\end{equation}
for a rather general class of utility functions $U$ (see Assumption 1 on $U$ below).

For convenience, for any given utility function $U$ we call the pair $(U, X)$ with $X$ given by (\ref{one}) an economy from now on. We say that the problem (\ref{new-L2}) has a solution if there exists a portfolio $x_0\in \R^d$ with the finite  Euclidean norm $|x_0|<+\infty$ and with 
$-\infty<EU(W(x_0))<+\infty$ such that 
\[
EU(W(x))\le EU(W(x_0))
\]
for all $x\in \R^d$. For a given economy $(U, X)$, existence of such a $x_0$ is not guaranteed even under the condition that the utility function $U$ is bounded, non-decreasing, and continuous, see Example \ref{ex-fi} below for this. Clearly, under these conditions on $U$ the map $x\rightarrow EU(W(x))$ is continuous but the portfolio space $\R^d$ is unbounded in the Euclidean norm.

\subsection{Well-posedness}

The aforementioned  facts illustrate the need for the introduction  of some conditions on $(U, X)$ that can guarantee the existence of a solution for (\ref{new-L2}). In fact, there are several questions that we need to address when we study the problem (\ref{new-L2}): (i) What are the sufficient conditions on $(U, X)$ that can guarantee $EU(W(x))<+\infty$ whenever  $|x|<+\infty$, (ii) What conditions on $(U, X)$  guarantee $\sup_{x\in \R^d}EU(W(x))<+\infty$, (iii) Is there a  portfolio $x_0\in \R^d$ with finite Euclidean norm $|x_0|<+\infty$ such that $EU(W(x_0))=\sup_{x\in \R^d}EU(W(x))$ holds. 

Before we discuss these problems we first present some examples. In the following example, we present an economy $(U, X)$ where $U$ is bounded, continuous, non-decreasing but there is no portfolio $x_0$ with a finite Euclidean norm such that
$EU(W(x_0))=\sup_{x\in \R^d}EU(W(x))$.

\begin{example} \label{ex-fi} Consider the model (\ref{one}) in dimension one, i.e., $d=1$.  Assume $\mu=0, \gamma=0, A=1, Z=1$, and the risk-free interest rate is zero $r_f=0$. The corresponding wealth  is given by $W(x)=W_0+xW_0N(0, 1)$, where $W_0>0$ is the initial wealth of the investor. For the utility function $U(w)$ take 
\begin{equation*}
U(w)=\left \{
\begin{array}{cc}
  m   &  w\geq m, \\
   w  &  0\le w \le m,\\
   0  &  w\le 0,
\end{array}
\right.
\end{equation*}
for some finite positive number $m$. This utility function is continuous, non-decreasing, and bounded.
Below we will show that $sup_{x\in \R}EU(W(x))=m/2$ and $EU(W(x))<m/2$ when $x$ is a finite number as long as $m>2W_0$. First observe that
\[
W(+\infty)=:\lim _{x\rightarrow +\infty} W(x)=\left\{
\begin{array}{cc}
+\infty  & N>0, \\
  -\infty   & N<0,
\end{array}
\right.
\]
and hence we have $EU(W(+\infty))=m/2$. Denote $G(x)=EU(W(x))$. We have $G(0)=U(W_0)$ and $G(x)=G(-x)$. Next we calculate $G(x)$ for $x>0$ explicitly. We have 
\[
G(x)=m+(W_0-m)\Phi(\frac{m-W_0}{xW_0})-W_0\Phi(-\frac{1}{x})-\frac{xW_0}{\sqrt{2\pi}}[e^{-\frac{(m-W_0)^2}{2x^2W_0^2}}-e^{-\frac{1}{2x^2}}], \; \; x>0.
\]
The first order derivative of $G$ equals to
\[
G'(x)=\frac{W_0}{\sqrt{2\pi}}e^{-\frac{1}{2x^2}}[1-e^{\frac{m}{x^2W_0}(1-\frac{m}{2W_0})}].
\]
If $m>2W_0$ then $G'(x)>0$ for all $x>0$.
Therefore when $m>2W_0$ the expected utility function  $G(x)$ is a strictly increasing function of the portfolio $x>0$. We have $\lim_{x\rightarrow +\infty}G(x)=m/2$ by the dominated convergence theorem (as the utility function is bounded). Recall that $G(0)=U(W_0)$ and under the condition  $m>2W_0$ we have $G(0)=W_0<m/2$. From  these we conclude that 
\[
\sup_{x\in \R}EU(W(x))=m/2,
\]
while for any finite number $x$ we have $EU(W(x))<m/2$. Thus in the economy $(U, X)$ in this example, to reach the maximum possible utility level $d/2$ one needs to buy infinite amount of the risky asset ($x=+\infty$) or short-sell infinite amount of the risky asset ($x=-\infty$). Observe here that the utility function $U$ is bounded, non-decreasing, continuous. Nonetheless, the optimal expected utility can be achieved only at the infinity portfolio.  
\end{example}

In the next Example, we present an economy $(U, X)$ where any long position on the risky asset results in $+\infty$ expected utility and any short  position on the risky asset results in $-\infty$ expected utility. Expected utility maximization problems in such economies clearly become meaningless.

\begin{example}\label{ex3.11} Consider the model (\ref{one}) in dimension $d=1$. Assume  $\mu>0, \gamma>0, A=1$ in this model and let $Z$ be any non-negative mixing random variable with $EZ=+\infty$  and $E\sqrt{Z}<+\infty$ (For an example $Z$ can be the lottery in the ``St. Petersburg Pardox'' that takes value $2^{k-1}$ with probability $\frac{1}{2^k}$ for each positive integer $k$). Take $r_f=0$ and let $U(x)=x$ be the utility function of the agent. The corresponding wealth in (\ref{wealth}) is 
\[
W(x)=W_0+W_0[x\mu+x\gamma Z+x\sqrt{Z}N(0, 1)],
\]
for any portfolio $x\in \R$. One can easily show that
\begin{equation}\label{ex3.1}
EU(W(x))=\left \{ \begin{array}{cc}
+\infty &  x>0,\\
 -\infty  & x<0. 
\end{array}
\right.
\end{equation}
We clearly have $EU(W(0))=W_0$. The relation (\ref{ex3.1}) shows that any short position on the risky asset gives $+\infty$ expected utility and any long position on the risky asset gives an expected utility that equals to  $-\infty$. 
\end{example}

These examples explain that some condition on the economy $(U, X)$ is necessary for the  problem (\ref{new-L2}) to be well-posed.

\begin{definition}\label{def3.1} We say that the utility maximization problem (\ref{new-L2}) is well-posed if there exists a portfolio 
$x^{\star}\in \R^d$ with $|x^{\star}|<+\infty$ such that 
\[
EU(W(x))\le EU(W(x^{\star}))<+\infty,
\]
for all $x\in \R^d$. If not then we call the problem (\ref{new-L2}) ill-posed.
\end{definition}

\begin{definition} We say that an economy $(U, X)$ admits asymptotically optimal portfolio (AOP) if there exists a sequence of portfolios $\{x_n\}$ with divergent Euclidean norm, i.e., $|x_n|\rightarrow +\infty$,  such that
\[
EU(W(x_n))\rightarrow \sup_{x\in \R^n}EU(W(x)),
\]
while there is no a portfolio $x_0$ with finite Euclidean norm with
\[
EU(W(x_0))=\sup_{x\in \R^n}EU(W(x)).
\]
\end{definition}

\begin{remark}We remark here that our above definition of well-posedness of the problem  (\ref{new-L2}) is in line with the definition  of well- posedness of the expected utility maximization problem under cumulative prospect theory utility functions that was discussed in the paper \cite{xun-yu-zhou} (see section 3 of this paper and  also see Proposition 1 of \cite{Pirvu_Kwak}). Note that in the definition of AOP above both of the cases 
$\sup_{x\in \R^n}EU(W(x))<+\infty$ and $\sup_{x\in \R^n}EU(W(x))=+\infty$ are allowed. The economy in Example \ref{ex-fi} above admits AOP while $\sup_{x\in \R^n}EU(W(x))<+\infty$ as the utility function in this example is a bounded function.
\end{remark}
\begin{remark}\label{rem3.6} If there exists a portfolio $x_0$ with finite Euclidean norm such that $EU(W(x_0))=\sup_{x\in \R^d}EU(W(x))$, then the problem (\ref{new-L2}) is well-posed. If not then,  since we always have a sequence $\{x_n\}$ of portfolios with $EU(W(x_n))\rightarrow \sup_{x\in \R^n}EU(W(x))$,  if $\{x_n\}$ contains a sub-sequence $\{x_{n_k}\}$ with divergent Euclidean norm the economy $(U, X)$ admits AOP.
The other possible case is $\{x_n\}$ is a bounded family in the Euclidean norm. In Lemma \ref{lem3.55} below we introduce the condition (\ref{the-con}) that rules out this last possibility in the economy $(U, X)$.
\end{remark}

From our above discussions it is clear that the first problem that one needs to address when  studying the problem (\ref{new-L2}) is if it is well-posed. These problems will be discussed in detail in the Section 4 below.  The  Lemma \ref{lem3.4} in this section clarifies some sufficient conditions on the utility function $U$ for the existence of a solution for (\ref{new-L2}). Based on this  Lemma we introduce the following conditions on the utility function $U$.

\noindent \textbf{Assumption 1:} The utility function $U:\R\rightarrow \R$ is a finite valued, continuous, non-constant, non-decreasing, bounded from above, and $lim_{w\rightarrow -\infty}U(w)=-\infty$.

\begin{remark} We remark here that similar conditions on the utility functions were discussed in the paper \cite{Rasonyi-Huy}, see Assumption 4.1 at page 687 of \cite{Rasonyi-Huy}.
    
\end{remark}

\subsection{Examples of utility functions that satisfy Assumption 1}

The conditions in Assumption 1 above on the utility functions are not strong conditions in fact. Below we write down some examples for utility functions that satisfy Assumption 1.
\begin{example}\label{short} Let $\ell$ be any convex nondecreasing function defined on the nonnegative real line with $\lim_{x\rightarrow +\infty}\ell(x)=+\infty$. Define
\begin{equation*}
 U(x)=-\ell(x^-),   
\end{equation*}
where $x^{-}$ is the negative part of the real number $x$. Then $U(x)$ satisfy Assumption 1. See Example 2.3 of \cite{touzi-bouchard} for the origin of this example. Also see section 2.2.2 of \cite{Vicky-Henderson-David-Hubson} for an example of a utility function associated with shortfall hedging which minimizes expected loss.     
\end{example}

\begin{example} (Henderson-Hubson ($H\&H$) utility) For any real number $\tau>0$ consider the following utility function
\[
U(x)=\frac{1}{\tau}(1+\tau x-\sqrt{1+\tau^2x^2}),
\]
see Section 2.2.2 of \cite{Vicky-Henderson-David-Hubson} for the origin of this utility function. This class of utility functions is bounded from above by $\frac{1}{\tau}$, strictly increasing, strictly concave, and $\lim_{x\rightarrow -\infty}U(x)=-\infty$. Hence they satisfy Assumption 2 for any $\tau>0$.    
\end{example}

Next we present another class of utility functions that satisfy these conditions. Before doing this, we first recall few definitions. The Arrow-Pratt measure of absolute risk-aversion of a utility function $U$ is defined by
\[
A(w)=-\frac{U^{''}(w)}{U'(w)},
\]
and the Arrow-Pratt measure of risk-tolerance is defined by $T(x)=\frac{1}{A(w)}=-\frac{U'(w)}{U^{''}(w)}$.

\begin{definition}\label{def-saha} A utility function $U: \R\rightarrow \R$ is of the SAHARA class with risk aversion parameter $a>0$, scale parameter $b>0$,  and threshold wealth $\delta \in \R$ if it's risk- tolerance is given by
\[
T(w)=\frac{1}{a}\sqrt{b^2-(w-\delta)^2}.
\]
\end{definition}
For the details of this class of utility functions see \cite{Antoon-Pelsser} and also see section 5 of \cite{Moris-S}. It is straightforward to recover the SAHARA utility functions up to affine transformations by using the risk-tolerance $T(w)$ above. We have
\begin{equation}\label{63}
U(w)=-\frac{1}{a^2-1}\frac{(w-\delta)+a\sqrt{b^2+(w-\delta)^2}}{[(w-\delta)+\sqrt{b^2+(w-\delta)^2}]^a},    
\end{equation}
when $a\neq 1$ and 
\begin{equation}\label{6444}
  U(w)=\frac{\ln((w-\delta)+\sqrt{b^2+(w-\delta)^2})}{2}+\frac{w-\delta}{2b^2}[\sqrt{b^2+(w-\delta)^2}-(w-\delta)],  
\end{equation}
when $a=1$. The derivative of $U(w)$ for both of the cases $a\neq 1$ and $a=1$ is given by
\begin{equation*}\label{65}
 U'(w)=\frac{1}{[(w-\delta)+\sqrt{b^2+(w-\delta)^2}]^a}=b^{-a}e^{-a\times arcsinh(\frac{w-\delta}{b})}>0.   
\end{equation*}

We would like to check if these class of utility functions satisfy Assumption 1 above. We do this in the following Example. 

\begin{example} \label{3.8}  We have
\begin{equation}\label{66}
 \lim_{w\rightarrow +\infty}U(w)=\left \{
\begin{array}{cc}
  0,   & \mbox{ if $a>1$},\\
  +\infty, & \mbox{if $a\in (0, 1]$}. 
\end{array}
\right.
\end{equation}  
and 
\begin{equation}\label{6777}
 \lim_{w\rightarrow -\infty}U(w)=-\infty.
 \end{equation} 

To see this without loss of any generality we can assume $\delta=0$ in (\ref{63}). When $a\neq 1$ dividing both denominator and numerator of (\ref{63}) by $w$ we obtain 
\[
U(w)=-\frac{1}{a^2-1}\frac{(1+a\sqrt{b^2/w^2+1})}{[w^{1-1/a}+\sqrt{b^2/w^{(2/a)}+w^{2-2/a}}]^a} 
\]
When $a>1$ the denominator of this expression goes to $+\infty$  when $w\rightarrow +\infty$ and it goes to zero when $a\in (0, 1)$ as $w\rightarrow +\infty$. Hence (\ref{66}) holds when $a>0, a\neq 1$. When $a=1$ one can use (\ref{6444}) to show the claim in (\ref{66}). To show (\ref{6777}) it is sufficient to show the following limit (we do not include the factor $-1/(a^2-1)$ here)
\begin{equation*}\label{6888}
U_{-\infty}=:\lim_{w\rightarrow +\infty}\frac{-w+a\sqrt{b^2+w^2}}{(-w+\sqrt{b^2+w^2})^a}    
\end{equation*}
equals to $+\infty$ when $a>1$ and it equals to $0$ when $a\in (0,1)$. The case $a=1$ needs to be treated differently by using (\ref{6444}). We have
\begin{equation*}
\begin{split}
U_{-\infty}=&\lim_{w\rightarrow +\infty}\frac{(a^2(b^2+w^2)-w^2)/(w+a\sqrt{b^2+w^2})}{[b^2/(w+\sqrt{b^2+w^2})]^a}\\
=&\lim_{w\rightarrow +\infty}\frac{[a^2b^2+(a^2-1)w^2][w+\sqrt{b^2+w^2}]^a}{b^{(2a)}(w+a\sqrt{b^2+w^2})}\\
=&\lim_{w\rightarrow +\infty}\frac{[(a^2b^2)/w+(a^2-1)w][w+\sqrt{b^2+w^2}]^a}{b^{(2a)}(1+a\sqrt{b^2/(w^2)+1})}.
\end{split}
\end{equation*}
Clearly, the numerator of the last expression converges to $+\infty$ when $a>1$ and it converges to $-\infty$ when $0<a<1$, showing (\ref{6777}) for $a\neq 1$. When $a=1$, it is straightforward to show $\lim_{w\rightarrow -\infty}U(w)=-\infty$ by using (\ref{6444}).
\end{example}

\begin{remark} \label{cor311} For the Sahara utility functions $U$ with $a>1$, the optimization problem (\ref{new-L2}) always has a solution, i.e., there exists a $x_0\in \R^d$ such that $EU(W(x_0))>-\infty$ and $EU(W(x))\le EU(W(x_0))$ for all $x\in \R^d$.  To see this, observe that by Example \ref{3.8}, Sahara utility functions with $a>1$ satisfy Assumption 1. Hence by Lemma \ref{lem3.4}, the problem (\ref{new-L2}) has a solution. In fact the solution for (\ref{new-L2}) for the case of Sahara utility functions with $a>1$ is unique. This is due to the strict concave property of Sahara class of utility functions. This fact will be explained in the next section.
\end{remark}

\subsection{Closed form solution when the utility function is concave}

The purpose of this section is to present a closed form solution for the problem (\ref{new-L2}) in an economy $(U, X)$ when the utility function $U$ is concave and satisfies Assumption 1. For the well-posedness of the problem (\ref{new-L2}), the return vector $X$ also needs to satisfy certain conditions. We introduce the following assumption on the model (\ref{one}). 

\noindent \textbf{Assumption 2:} The model (\ref{one}) is such that $Z$ is strictly positive,  $EZ\in L^k$ for some positive integer $k$,  and $\mu-\1r_f+\gamma EZ\neq 0$.
\vspace{0.1in}

\begin{remark} The most popular mixing distribution $Z$ for the model (\ref{one}) is probably the class of generalized inverse Gaussian (GIG) random variables with density functions
\begin{equation}\label{GIG-PDF}
        f_{GIG}(x;\lambda,a, b) = \Big(\frac{b}{a} \Big)^{\frac{\lambda}{2}}\frac{x^{\lambda -1}}{2K_{\lambda}(\sqrt{ab})}e^{-\frac{1}{2}(a x^{-1} + b x)}\mathbf{1}_{(0, \infty)}(x), 
    \end{equation}
    where $K_{\lambda}(\cdot)$ is the modified Bessel function of the third kind, and the parameters $(\lambda, a, b)$ satisfy (i)  $a \geq 0, \quad b>0,  \text{ if } \lambda>0$, (ii) $a > 0, \quad b> 0, \text{ if } \lambda =0$, (iii) $a > 0, \quad b \geq 0, \text{ if } \lambda < 0$. The moments of these distributions are calculated at page 11 of \cite{Hammerstein_EAv_2010}. Based on this, all the GIG models with $\lambda \in \R, a>0, b>0$ satisfy the Assumption 1. This includes the inverse Gaussian random variable that corresponds to $\lambda=-\frac{1}{2}$. Therefore the popular Normal inverse Gaussian random variables satisfy Assumption 1. A gamma distribution corresponds to the parameters $\lambda>0, a=0, b>0,$ and in this case $EZ^r$ exists for all $r>-\lambda$. Therefore variance-gamma models also satisfy Assumption 1. The inverse gamma distribution corresponds to the parameter range $\lambda<0, a>0, b=0$. In this case $EZ^r<\infty$ for all $r<-\lambda$. Therefore inverse gamma mixing distributions with $\lambda<-1$ clearly satisfy Assumption 1. This means that skewed $t-$distributions with certain parameter ranges also satisfy Assumption 1.
   \end{remark}

Our main goal in this section is to provide some characterizations of the optimal portfolios for the problem (\ref{new-L2}) in an economy $(U, X)$ when $U$ satisfies Assumption 1 and $X$ satisfies Assumption 2. The solution for the problem (\ref{new-L2}) is not guaranteed to be unique. We will clarify in our discussions that if the utility function $U$ is strictly concave, then the solution for (\ref{new-L2}) is unique. We  also provide some characterizations of all the optimal portfolios for the problem  (\ref{new-L2}) when $U$ is merely concave.

First we recall some definitions. A random variable $W_1$ first-order stochastically dominates  another random variable $W_2$, denoted $W_1\succeq_1 W_2$, if it satisfies $EU(W_1)\geq EU(W_2)$ for every increasing function $U$ for which the two expectations are well defined. A random variable 
$W_1$ second-order stochastically dominates $W_2$, denoted $W_1\succeq_2 W_2$, if $EU(W_1)\geq EU(W_2)$ for every concave increasing function $U$ for which the two expectations are well defined. If the utility function $U$ is increasing and concave then for any two portfolios $y_1, y_2,$ if $W(y_1)\succeq_2 W(y_2)$ then we have $EU(W(y_1))\geq EU(W(y_2))$. In the proofs of this Section we need to use a result in \cite{Hasan}. We write down this  for convenience here.
\begin{proposition}\label{hes214} 
Given any model (\ref{one}) with $\mu+\gamma EZ\neq 0$. For any two portfolios $x_1, x_2\in \R^d$, the following relation
\[
x_1^T\mu+x_1^T\gamma EZ=x_2^T\mu+x_2^T\gamma EZ, \; \; x_1^T\Sigma x_1<x_2^T\Sigma x_2,
\]
implies $x_1^TX\succeq_{2} x_2^TX$.
  \end{proposition}

In the next Lemma we give characterizations of the optimal portfolios for the problem (\ref{new-L2}).

\begin{lemma}\label{main-th} Consider the optimization problem (\ref{new-L2}). Assume $U$ satisfies Assumption 1  and the model (\ref{one}) satisfies Assumption 2 above. Assume, in addition, $U$ is concave. Let $x\in R^d$ be any solution for (\ref{new-L2}). Then $x$ solves the following quadratic optimization problem 
 \begin{equation}\label{Qopt-2}
\begin{split}
 \min_{x} \; \;  &  \frac{1}{2}x^T\Sigma x, \\
s.t.\;\;  &x^T(\mu-r_f\1+\gamma EZ)=c,\\
\end{split}
\end{equation}   
for some real-number $c\geq 0$. This solution of (\ref{Qopt-2}), which we denote by $x_c$,  is given by
\begin{equation}\label{xstar}
x_c=\frac{c}{v^T\Sigma^{-1}v}\Sigma^{-1}v,
\end{equation}
where $v=:\mu-r_f\1+\gamma EZ$. With this optimal portfolio $x_c$ we have 
\begin{equation}\label{74}
 W(x_c)\overset{d}{=}W_0(1+r_f)+cW_0[\alpha +\beta Z+\sigma\sqrt{Z}N(0, 1)]
 \end{equation}
with 
\begin{equation}\label{etac}
\alpha=\frac{v^T\Sigma^{-1}(\mu-\1r_f)}{v^T\Sigma^{-1}v}, \; \; \beta=\frac{v^T\Sigma^{-1}\gamma}{v^T\Sigma^{-1}v}, \; \; \sigma=\frac{1}{\sqrt{v^T\Sigma^{-1}v}}.
\end{equation}
\end{lemma}
\begin{proof} First, from Lemma \ref{lem3.4} we know that the problem (\ref{new-L2}) has a solution. To see that the solution satisfies (\ref{Qopt-2}), observe that $W(x)=W_0(1+r_f)+W_0[x^T(\mu-\1 r_f)+x^T\gamma Z+\sqrt{Z}x^T\Sigma xN(0, 1)]$ and $EW(x)=W_0(1+r_f)+W_0[x^T(\mu-\1 r_f)+x^T\gamma EZ]$. Since $U$ is concave we have
\begin{equation}\label{7999}
 EU(W(x))\le U(EW(x))=U(x^T(\mu-\1 r_f)+x^T\gamma EZ).   
\end{equation}
If a portfolio $\bar{x}$ satisfies $\bar{x}^T(\mu-\1 r_f)+\bar{x}^T\gamma EZ<0$, then from (\ref{7999}) we see that $EU(W(\bar{x}))\le EU(W(0))$ (as the utility function $U$ is non-decreasing) and hence $\bar{x}$ can not be the optimal portfolio for (\ref{new-L2}). Therefore all the optimal portfolios $x$ for (\ref{new-L2})  should satisfy $x^T(\mu-\1 r_f)+x^T\gamma EZ\geq 0$. Now for each fixed $c\geq 0$ consider the problem (\ref{Qopt-2}). Let $x^{\star}$ be the solution of this problem. Then from Assumption 2 we have $W(x^{\star})\succeq_2 W(x)$ for all $x$ with $x^T(\mu-r_f\1+\gamma EZ)=c$. Then since $U$ is concave we have $EU(W(x^{\star}))\geq EU(W(x))$ when $x^T(\mu-r_f\1+\gamma EZ)=c$. We apply the Lagrangian method to (\ref{Qopt-2}) and obtain (\ref{xstar}). Then we plug (\ref{xstar}) into $W(x)$ and obtain (\ref{74}). This completes the proof.
\end{proof}
\begin{lemma}\label{3.122} Assume the model (\ref{one}) satisfies Assumption 2. Then for $\alpha$ and $\beta$ defined  in (\ref{etac}), we  have
\begin{equation*}
 \alpha+\beta EZ=W_0.   
\end{equation*}
\end{lemma}
\begin{proof} Let $v$ be defined as in Lemma \ref{main-th} above. We have
\[
v^T\Sigma^{-1}v=(\mu-r_f\1)^T\Sigma^{-1}(\mu-r_f\1)+2(\mu-r_f\1)^T\Sigma^{-1}\gamma EZ+\gamma^{-1}\Sigma^{-1}\gamma (EZ)^2>0,
\]
as $\Sigma^-$ is positive definite. Also by using (\ref{etac}) we obtain
 \[
 \begin{split}
\alpha+\beta EZ&=\frac{W_0}{v^T\Sigma^{-1}v}[(\mu-r_f\1)^T\Sigma^{-1}(\mu-r_f\1)+2(\mu-r_f\1)^T\Sigma^{-1}\gamma EZ+\gamma^{-1}\Sigma^{-1}\gamma (EZ)^2]\\
&=W_0.
 \end{split}\]
 This completes the proof.
\end{proof}

We observe that the random variable 
\begin{equation}\label{eta}
\eta=: \alpha+\beta Z+\sigma \sqrt{Z}N(0, 1),
\end{equation}
in (\ref{74}) is not related with the parameter $c$. For convenience we introduce the following notation
\begin{equation*}\label{wc}
 \kappa(c)=:W_0(1+r_f)+W_0c\eta,   
\end{equation*}
where $\eta$ is given by (\ref{eta}). 
We define the following function
\begin{equation}\label{fff}
\begin{split}
\Gamma(c)=:
EU\big [\kappa(c) \big ], \; \; c\geq 0, 
 \end{split}
 \end{equation}
and we observe that $\Gamma(0)=U(W_0(1+r_f))$.

In the next Lemma we study some properties of the function $\Gamma(c)$ defined in (\ref{fff}).
\begin{lemma}\label{lem311} Assume the utility function $U$ satisfies Assumption 1 and the model (\ref{one}) satisfies Assumption 2 . Let $\hat{c}\in [0, +\infty)$ be any number such that $\Gamma(\hat{c})>-\infty$, where $\Gamma (c)$ is given by (\ref{fff}). Then we have the following.
\begin{enumerate}
\item[i)]  If $U$ is  concave  then the function $\Gamma(c)$  satisfies $\Gamma(c)>-\infty$ for all $c\in [0, \hat{c}]$ and it is concave  on $[0, \hat{c}]$. If $U$ is strictly concave, then $\Gamma(c)$ is strictly concave on $[0, \hat{c}]$ as well.
\item [ii)] We have $\lim_{c\rightarrow +\infty}\Gamma(c)=-\infty$.  

\item [iii)] $\Gamma(c)$ is upper semi-continuous on $[0, \hat{c}]$.
\end{enumerate}
\end{lemma}
\begin{proof} 
i) Take any $c_1, c_2\in [0, \hat{c}]$ and $\lambda \in [0, 1]$. Observe that $\kappa(\lambda c_1+(1-\lambda c_2))=\lambda \kappa(c_1)+(1-\lambda)\kappa(c_2)$. Hence we have
\[
\Gamma(\lambda c_1+(\1-\lambda)c_2)=EU(\lambda \kappa(c_1)+(1-\lambda)\kappa(c_2)).
\]
Since $U$ is concave we have $\lambda U(\kappa(c_1))+(1-\lambda)U(\kappa(c_2))\le U(\lambda \kappa(c_1)+(1-\lambda)\kappa(c_2))$. The conditions $c_1, c_2\in [0, \hat{c}], \lambda\in [0, 1]$,  imply that all of $EU(\kappa(c_1)), EU(\kappa(c_2)), EU(\kappa(\lambda \kappa(c_1)+(1-\lambda)\kappa(c_2)))$ are finite numbers. Therefore we have $\Gamma(\lambda c_1+(1-\lambda c_2))\geq \lambda \Gamma(c_1)+(1-\lambda)\Gamma(c_2)$ showing that $\Gamma(c)$ is concave on $[0, \hat{c}]$. The claim that $\Gamma(c)$ is strictly concave on $[0, \hat{c}]$ if $U$ is strictly concave follows from above analysis easily.

ii) Observe that
\[
\Gamma(c)=EU(\kappa(c))=E[U(\kappa(c))1_{\eta\geq 0}]+E[U(\kappa(c))1_{\eta<0}],
\]
and each of the events $\{\eta\geq 0\}$ and $\{\eta<0\}$ has positive probability under Assumption 2. Since $U$ is bounded from above, 
$E[U(\kappa(c))1_{\eta\geq 0}]$ is bounded above. So it is sufficient to show that 
$\lim_{c\rightarrow +\infty}E[U(\kappa(c))1_{\eta<0}]=-\infty$. To see this observe that $\lim_{c\rightarrow +\infty}\kappa(c)=-\infty$ on the event $\{\eta<0\}$. Therefore, due to Assumption 1, we have $\lim_{c\rightarrow +\infty}U(\kappa(c))=-\infty$ on $\{\eta<0\}$. By using Fatou's Lemma we have 
\begin{equation*}
\begin{split}
-\limsup_{c\rightarrow +\infty}E[U(\kappa(c)1_{\eta<0})]=  &\liminf_{c\rightarrow +\infty}E[-U(\kappa(c)1_{\eta<0})]\geq E\liminf_{c\rightarrow +\infty}[-U(\kappa(c))1_{\eta<0}]\\
=&-E\limsup_{c\rightarrow +\infty}[U(\kappa(c))1_{\eta<0}]=+\infty.
\end{split}
\end{equation*}
This shows that $\lim_{c\rightarrow +\infty}E[U(\kappa(c)1_{\eta<0})]=-\infty$. From this the claim in ii) follows.

iii)  Take any $c_0\in [0, \hat{c}]$. If $c_0=0$ the limit $\lim_{c\rightarrow c_0}$ is understood from the right-hand-side and if $c_0=\hat{c}$ the limit $\lim_{c\rightarrow c_0}$ is understood from the left-hand-side in the below discussions. First observe that $\lim_{c\rightarrow c_0}\kappa(c)=\kappa(c_0)$ almost surely. Since $U$ is continuous by assumption $\lim_{c\rightarrow c_0}U(\kappa(c))=U(\kappa(c_0))$ almost surely. By Fatou's Lemma we have 
\begin{equation*}
\begin{split}
-\limsup_{c\rightarrow c_0}\Gamma(c)=&\liminf_{c\rightarrow c_0}E[-U(\kappa(c))]\geq E\liminf_{c\rightarrow c_0}[-U(\kappa(c))]\\
=&-E[U(\kappa(c_0))]=\Gamma(c_0),
\end{split}
\end{equation*}
which shows $\limsup_{c\rightarrow c_0}\Gamma(c)\le \Gamma(c_0)$.
\end{proof}
\begin{remark} The upper semi-continuity of $\Gamma(c)$ shows that $\Gamma(c)$ has global maximum on any compact subset of $[0, +\infty]$. Hence the problem (\ref{F-nu}) is well-defined. Recall that the class of Sahara utility functions with parameters $a>1$, satisfy $\lim_{w\rightarrow -\infty}U(w)=-\infty$ as shown in Lemma \ref{3.8} above. Thus for Sahara utility functions with $a>0$, the corresponding $\Gamma(c)$ satisfy $\Gamma(0)=U(W_0(1+r_f))$ (a finite value) and $\lim_{w\rightarrow +\infty}\Gamma(c)=-\infty$. Also  $\Gamma$ is strictly concave on $[0, \bar{c}]$ for any $\bar{c}$ with $\Gamma(\bar{c})>-\infty$, as Sahara utility functions are strictly concave. 
\end{remark}

 The next result establishes a relation of the function in (\ref{fff}) with the optimization problem (\ref{new-L2}) above. Especially this Proposition shows that the solution for the optimization problem (\ref{new-L2}) is always unique when $U$ is strictly concave. 
\begin{theorem}\label{3.155} Assume the utility function $U$ satisfies Assumption 1 and the model (\ref{one}) satisfies Assumption 2. Assume, in addition, $U$ is concave. Then a portfolio $x^{\star}$ is an optimal portfolio for (\ref{new-L2}) if and only if 
\begin{equation}\label{main-xstar}
x^{\star}=\frac{c^{\star}}{v^T\Sigma^{-1}v}\Sigma^{-1}v,
\end{equation}
where $v=:\mu-r_f\1+\gamma EZ$ and 
\begin{equation}\label{F-nu}
 c^{\star}\in \arg max_{c\in [0, +\infty)}\Gamma(c),     
\end{equation}
with $\Gamma(c)$ given in (\ref{fff}). If $U$ is strictly concave then the solution for (\ref{new-L2}) is unique.
\end{theorem}
\begin{proof} Assume $x^{\star}$ is an optimal portfolio for (\ref{new-L2}). Then by Lemma \ref{main-th}, $x^{\star}$ solves (\ref{Qopt-2}) for some $c=\bar{c}\geq 0$. The solution of (\ref{Qopt-2}) with $c$ replaced by $\bar{c}$ is given by
\[
x^{\star}=x_{\bar{c}}=\frac{\bar{c}}{v^T\Sigma^{-1}v}\Sigma^{-1}v.
\]
From (\ref{74}) we have
\[
EU(W(x^{\star}))=\Gamma(\bar{c}).
\]
Now any portfolio $x_c$ of the form (\ref{xstar}) with any $c\geq 0$ satisfies $EU(W(X_c))=\Gamma(c)$ again due to (\ref{74}). The optimality of $x^{\star}$ then gives 
$\Gamma(\bar{c})\geq \Gamma(c)$ for any $c\geq 0$. This shows that $\bar{c}\in \arg max_{c\in [0, +\infty)}\Gamma(c)$.

Now assume $x^{\star}$ is given by (\ref{main-xstar}) with $c^{\star}\geq 0$ given by (\ref{F-nu}). Then from (\ref{74}) we have $\Gamma(c^{\star})=EU(W(x^{\star}))$. For any other arbitrarily fixed portfolio $x_0\in \R^d$, denote $c_0=:x^T_0v$. Let $\bar{x}$ be the solution of (\ref{Qopt-2}) with $c$ replaced by $c_0$. Then we have $EW(x_0)=EW(\bar{x})$ while $\bar{x}^T\Sigma^{-1} \bar{x}\le \bb x_0^T\Sigma^{-1} x_0$. Hence from Assumption 2 we have $W(\bar{x})\succeq_2 W(x_0)$ implying $EU(W(\bar{x}))\geq EU(W(x_0))$. Now if $c_0<0$, then $EW(\bar{x})=W_0(1+r_f)+W_0c_0<W_0(1+r_f)$ and since $U$ is concave and non-decreasing we have $EU(W(\bar{x}))\le U(EW(\bar{x}))\le U(W_0(1+r_f))=EU(W(0))$. So any portfolio  $x_0$ with $c_0=x_0^Tv<0$ can not be optimal. Now take any portfolio $x_0$ with $c_0=x_0^Tv\geq 0$. Define 
\[
\bar{x}=:\frac{c_0}{v^T\Sigma^{-1}v}\Sigma^{-1}v.
\]
Then $\bar{x}$ solves the problem (\ref{Qopt-2}) with
$c$ replaced by $c_0$. Hence by Assumption 2 we have $W(\bar{x})\succeq_2 W(x_0)$. At the same time, from (\ref{74}) we have $\Gamma(c_0)=EU(W(\bar{x}))$. From these we conclude 
\[
EU(W(x^{\star}))=\Gamma(c^{\star})\geq \Gamma(c_0)=EU(W(\bar{x}))\geq EU(W(x_0)).
\]
 as $c^{\star}$ satisfies (\ref{F-nu}). Hence $x^{\star}$ is the optimizing portfolio for the problem (\ref{new-L2}).
\end{proof}
\begin{remark} As Theorem \ref{3.155} shows the optimal portfolio for the optimization problem (\ref{new-L2}) is related with the maximum value of the function $\Gamma(c)$ on $[0, +\infty)$.
For any two $c_1\geq 0, c_2\geq 0$ with $c_1\geq c_2$, the random variable $w(c_1)$ has higher mean than the random variable $w(c_2)$, i.e.,  $Ew(c_1)=W_0(1+r_f)+c_1\eta +c_1\kappa EZ \geq W_0(1+r_f)+c_2\eta +c_2 \kappa EZ=Ew(c_2)$ due to Lemma \ref{3.122}. But at the same time $c_1\sigma>c_2\sigma$, where $\sigma$ is given by  (\ref{etac}). Hence we can not claim, by using the Assumption 2 above, that  $w(c_1)\succeq_2 w(c_2)$ which would have implied $\Gamma(c_1)\geq \Gamma(c_2)$ as $U$ is concave.  
\end{remark}

\begin{example}\label{log_normal} Consider an economy with $d+1$ assets. The risk-free asset dynamics is $dB_t=r_fB_tdt$ and the remaining $d$ assets follow multi-dimensional Black-Scholes model:
\[
\frac{dS_t^{(i)}}{S_t^{(i)}}=\beta_idt+\sigma_idW_t^{(i)}, i=1, 2, \cdots, d,
\]
where $\beta_i$ represents the drift rate of stock $i$, $W_t^{(i)}$ is standard Brownian motion, and $\sigma_i$ is the volatility of the stock $i$. The Wiener processes $W_t^{(i)}$ and   $W_t^{(j)}$  are correlated with correlation coefficient $\rho_{ij}$. Denoting $\mu=(\mu_i)_{1\le i\le d}$ where $\mu_i=\beta_i-\frac{1}{2}\sigma_i^2, i=1, 2, \cdots, d$, the one-period log-returns of the $d$ risky assets, which we denote by $X$,  in this economy satisfies 
\[
X\overset{d}{=}\mu+\Sigma^{\frac{1}{2}}N_d,
\]
where $\Sigma=(\sigma_i\sigma_j\rho_{i,j})$ is the co-variance matrix (which we assume strictly positive definite) and $N_d$ is a $d-$dimensional standard Normal random vector. Now take any utility function $U$ that satisfies the Assumption 1 and consider the problem (\ref{new-L2}). From Theorem
\ref{3.155}, the corresponding optimal portfolios take the following form
\[
x^{\star}=\frac{c^{\star}}{v^T\Sigma^{-1} v}\Sigma^{-1}v,
\]
where $v=\mu-r_f\1$ and $c^{\star}=\arg max _{c\in [0, +\infty)]}\Gamma(c)$ with
\[
\Gamma(c)=EU[W_0(1+r_f)+W_0c+W_0c N(0, 1)].
\]
As Lemma \ref{lem311} shows the function $\Gamma(c)$ is concave, upper semi-continuous, and $\lim_{c\rightarrow +\infty}\Gamma(c)=-\infty$. These conditions imply that $c^{\star}$ exists and finite. We remark here that the conclusion of this example follows easily by using the stochastic dominance relation among Normal random variables as our utility functions are concave. Here we merely presented this simple example as a direct application of our Theorem \ref{3.155}. 

\end{example}

\begin{example} Following the idea of Example \ref{short}, take $\ell(x)=x^{q}, q>1, x\geq 0,$  and define $U(x)=-(x^-)^q$. Then the function $\Gamma $ in (\ref{fff}) is given by
\[
\Gamma(c)=-E[(w(c))^-]^q, \; c\geq 0.
\]
By Lemma \ref{lem311} this function is a concave function and $\lim_{c\rightarrow +\infty}\Gamma(c)=-\infty$. Since the model satisfies Assumption 2,   $Z$ is in $L^k$ for some $k\geq 1$ and  hence  $w(c)$ is integrable for each $c\geq 0$. These do not guarantee that $(w(c))^-\in L^q$ for a given $q>1$ and in the case of $(w(c))^-\notin L^q$ for some $c\in [0, +\infty)$  we have $\Gamma(c)=-\infty$ for such $c$. By Theorem \ref{3.155}, the corresponding expected utility maximizing portfolios are given by (\ref{main-xstar}). The $c^{\star}$ here are the maximizing points (if they are not unique) of $\Gamma$ on the non-negative real line and they need to be found by using numerical procedures.    
\end{example}

\subsection{Two-fund separation}

The two-fund separation theorem, introduced by \citet{Tobin58}, is a cornerstone of modern portfolio theory. It states that an investor with a quadratic utility function should divide their initial wealth allocation into two distinct steps. First, they should identify the tangency portfolio -- the combination of risky assets that maximizes the Sharpe ratio. Then, they should decide on the optimal mix between this tangency portfolio and the risk-free asset, depending on the investor’s attitude toward risk
\footnote{Confusingly, the two-fund separation theorem (with a risk-free asset) is often referred to as “the one-fund theorem” in the literature (e.g., \citet{Rockafellar.etal06}). According to \citet{Ross78},  who defined “(strong/weak) $k$-fund separability”, the two-fund separation theorem (where all risk averse investors choose portfolios made up of investment in a fund consisting of only the risk-free asset and the tangency portfolio of risky assets) corresponds to {\it one-fund separability}.}.

Originally, Tobin’s two-fund separation theorem, together with Markowitz's mean-variance analysis, was formulated on the assumptions that asset returns are normally distributed and investors have quadratic utility functions. Over the decades, great efforts have been made to relax these assumptions. \citet{Cass.Stiglitz70} extended Tobin’s work by showing that the separation property holds for more general utility functions, specifically those that exhibit constant relative risk aversion (CRRA). \citet{Owen-Rabinovitch-1983} demonstrated that the two-fund separation property holds for a broader class of distributions beyond the normal distribution, specifically elliptical distributions.


Moreover, since empirical studies indicate that asset returns are skewed in addition to fat-tailedness, suitable classes of probability distributions to capture such ``stylized facts'' have been sought, together with mild assumptions about investors’ preferences, see, e.g., \citet{Mencia.Sentana09}, or more recently, \citet{Birge_Chavez_2016, Birge-Bedoya},  \citet{Vanduffel.Yao17},  \citet{Bernard.etal21}, to name a few.

Our paper contributes to this line of research by extending the two-fund separation theorem from quadratic utility functions to a broad class of utility functions, albeit at the cost of limiting return vectors to follow a particular class of distributions. As mentioned earlier, our work shares an essentially common distributional setup with \citet{Birge_Chavez_2016, Birge-Bedoya},  \citet{Vanduffel.Yao17}, namely NMVM models.


Now let $\delta$ denote the proportion of initial wealth $W_0$ invested in the riskless asset and let $b_i, 1\le i\le d,$ denote the proportion of the remainder $W_0(1-\delta)$ invested in the $i-$th risky asset. Here $\delta$ is allowed to take negative values which corresponds to holding short positions on the risk-free asset. Then the corresponding wealth is given by
\begin{equation*}
\begin{split}
\Bar{W}(\delta, \bar{b})=&W_0[1+\delta r_f+(1-\delta)\bar{b}^TX],\\
=&W_0[1+r_f+(1-\delta)\bar{b}^T(X-\1r_f).
\end{split}
\end{equation*}
where $\bar{b}=(\bar{b}_1, \bar{b}_2, \cdots, \bar{b}_d)^T$. We clearly have $W(x)=\bar{W}(\delta, \bar{b})$ as long as $\sum_{i=1}^dx_i\neq 0$ and $\bar{b}=x/(\sum_{i=1}^dx_i), \; \delta=1-\sum_{i=1}^dx_i$. Observe that a portfolio $x$ with $\sum_{i=1}^dx_i=0$ corresponds to $\delta=1$,  investing all the initial wealth $W_0$ on the risk-free asset.

\begin{proposition}\label{t-fund} Consider the optimization problem (\ref{new-L2}) for a given economy $(U, X)$. Assume the return vector $X$ satisfies Assumption 2 and the utility function $U$ satisfies Assumption 1 and it is strictly concave. Denote by
$(b_i)_{1\le i\le d}$ the components of the vector $\Sigma^{-1}v$ where  $v=\mu-r_f\1+\gamma EZ$. If $b_0=:\sum_{i=1}^db_i\neq 0$ let $\bar{b}=\frac{1}{b_0}\Sigma^{-1}v$. Then the optimal investment strategy in the economy $(U, X)$ is a combination of the risk-free asset $r_f$ and the mutual fund $\bar{b}$.
\end{proposition}
\begin{proof} From Theorem \ref{3.155}, the optimal portfolio is given by $x^{\star}=\frac{c^{\star}}{v^T\Sigma^{-1}v}\Sigma^{-1}v$ and it is unique. If $\sum_{i=1}^dx_i^{\star}=0$, which can happen if $c^{\star}=0$ for example, then $\delta^{\star}=1-\sum_{i=1}^dx_i^{\star}=1$ and in this case the optimal portfolio is to invest all the initial wealth $W_0$ on the risk-free asset. If $\sum_{i=1}^dx_i^{\star}\neq 0$ then clearly $b_0\neq 0$. In this case we define $\bar{b}=\frac{1}{b_0}\Sigma^{-1}v$. Hence when $\sum_{i=1}^dx_i^{\star}\neq 0$ the utility maximizing optimal strategy in the economy is to invest $\delta^{\star}=1-\sum_{i=1}^dx_i^{\star}$ proportion of the initial wealth $W_0$ into the risk-free asset and the proportion $1-\delta^{\star}$ of the initial wealth on the mutual fund $\bar{b}$. Here this proportion $\delta^{\star}$ depends on $c^{\star}$ and hence on the initial wealth $W_0$ and the utility function $U$ that the individual employs (assume that the model (\ref{one}) is fixed). 
\end{proof}

\begin{example} Consider the optimization problem (\ref{new-L2}) with the SAHARA utility function $U$ with parameter $a>0$. See Example \ref{3.8} for the details of this utility function. Let $X$ be any given model (\ref{one}) with $Z$ being any non-trivial non-negative random variable. Then according to Proposition \ref{t-fund}, the optimal strategy for the investor is to divide his wealth between the risk-free asset and the mutual fund $\bar{b}$ defined in the Proposition \ref{t-fund}.
\end{example}

\section{Numerical illustrations}

In this Section we numerically illustrate our results in the earlier Sections. First recall that for any square integrable random vector $Y$ with mean vector $m=EY$ and covariance matrix $\Omega=Cov(Y)$, the associated tangent portfolio takes the following form
\begin{equation}\label{tangent}
\mathcal{T}_{mv}=\frac{1}{\1^T\Omega^{-1}(m-\mathbf{1}r_f)}\Omega^{-1}(m-\mathbf{1}r_f).    
\end{equation}
To make comparisons between tangent portfolios  (\ref{tangent}) and optimal portfolios $x^{\star}$ in our Theorem \ref{3.155} we first normalize $x^{\star}$ so that the components of the normalized one adds up to one. We have \[
\1^Tx^{\star}=\frac{c^{\star}}{v^T\Sigma^{-1}v}\1^T\Sigma^{-1}v, \;\;\; \bar{x}^{\star}=:\frac{x^{\star}}{\1^Tx^{\star}}=\frac{1}{\1^T\Sigma^{-1}v}\Sigma^{-1}v.
\]
Clear $\bar{x}^{\star}$ is the tangent portfolio that corresponds to a return vector $\bar{Y}$ with mean $E\bar{Y}=\mu+\gamma EZ$ and covariance matrix $Cov(\bar{Y})=\Sigma$ (recall hear $\Sigma=A^TA$). We call this  ``skewness-induced tangent portfolio''  and we denote it by $\mathcal{T}_{skew}$. We remark here that, in comparison, the mean-variance frontier tangent portfolio of the vector $Y=X$ takes the form (\ref{tangent}) with $m=\mu+\gamma EZ$ and $\Omega=Var (Z)\gamma^T\gamma+EZ \Sigma$.

By using the tangent portfolio $\mathcal{T}_{skew}$, the optimal portfolio $x^{\star}$ can be written as 
\[
x^{\star}=\frac{\1^T\Sigma^{-1}v}{v^T\Sigma^{-1}v}c^{\star}\mathcal{T}_{skew},
\] 
where $c^{\star}$ is computed by (\ref{F-nu}).  The constant $\frac{\1^T\Sigma^{-1}v}{v^T\Sigma^{-1}v}c^{\star}$ depends on the utility function $U$, the return vector $X$, the risk-free interest rate $r_f$, and the initial wealth $W_0$. We use the notation $\lambda_{U}=\lambda(U, r_f, X, W_0)$ to denote it, i.e., $\lambda_{U}=\frac{\1^T\Sigma^{-1}v}{v^T\Sigma^{-1}v}c^{\star}$. 

The financial interpretation of the portfolio $x^{\star}$ is that an expected utility maximizer with utility function $U$ invests the proportion $\lambda_U$ of the initial wealth $W_0$ on the tangent portfolio $\mathcal{T}_{skew}$ and the remaining proportion $1-\lambda_U$ on the risk-free asset with return $r_f$ as optimal strategy. The components of the tangent portfolio $\mathcal{T}_{skew}$ represent the proportions that needs to be invested optimally on each risky asset. Clearly the tangent portfolio $\mathcal{T}_{skew}$ is independent from the utility function $U$ and the initial wealth $W_0$ and it is determined by the return vector $X$ and the risk-free rate $r_f$. Hence the effect of the utility function $U$ and the initial wealth $W_0$ on the optimal investment is reflected on the proportion $\lambda_U$ that one needs to invest on the risky assets only. 

From the above discussions it becomes clear that our numerical tests on our Theorem \ref{3.155} boils down to computing $\lambda_U$ for different utility functions and different models (\ref{one}). For our numerical tests we consider skewed $t-$distributions which are given by the model (\ref{one}) with $Z\sim InverseGamma( \delta /2, \delta /2), \delta >2$. With this model we have 
\begin{equation}\label{skew}
v=\mu+\frac{\delta}{\delta-2}\gamma-\1 r_f, \; \;  \mathcal{T}_{skew}=\frac{x^{\star}}{\1^Tx^{\star}}=\frac{1}{\1^T\Sigma^{-1}v}\Sigma^{-1}v.
\end{equation}
The vector $\mathcal{T}_{skew}$ already tells us the optimal proportion that needs to be invested on each risky assets. Therefore it remains to calculate $\lambda_U$. It depends on the utility function of the investor. For our numerical tests we consider the following family of SAHARA utility functions

\begin{equation*}
U(a, b)=U(w; a, b)=-\frac{1}{a^2-1}\frac{w+a\sqrt{b^2+w^2}}{[w+\sqrt{b^2+w^2}]^a}, \;\;\;    a>1, b>0.
\end{equation*}
To reflect the dependence on the parameters of these type of utility functions we denote $\lambda_{U(a, b)}$. Recall that $\lambda_{U(a, b)}=\frac{c^{\star}(a, b)}{v^T\Sigma^{-1}v}\1^T\Sigma^{-1}v$ with 
\[
c^{\star}(a, b)=\arg max_{c\in [0, \infty)} \Gamma(c),
\]
where the function $\Gamma(c)$ is given by (\ref{fff}). 

For our model (\ref{one}), we estimate its parameters from the real work data. We consider the daily log-returns of \texttt{Apple Inc. (AAPL)}, \texttt{Microsoft Corporation (MSFT)}, \texttt{Alphabet Inc. (GOOGL)}, and \texttt{Amazon.com, Inc. (AMZN)} from January 1, 2020 to January 1, 2021. Below table is a summary of these 4 stocks.
\begin{table}[H]
\centering
\caption{Descriptive statistics of stock returns}
\label{tab:stock-stats}
\begin{tabular}{lcccc}
\hline \hline
Stock & Mean & Variance & Skewness & Kurtosis \\
\hline 
AAPL  & 0.002293 & 0.000866 & -0.301542 & 6.818200 \\
AMZN  & 0.002143 & 0.000586 & -0.018665 & 4.222902 \\
GOOGL & 0.000981 & 0.000592 & -0.473868 & 7.205632 \\
MSFT  & 0.001334 & 0.000770 & -0.418345 & 9.770342 \\
\hline \hline
\end{tabular}
\end{table}

The annual risk-free rate is taken to be $1.25\%$. We apply the Expectation–Maximization (EM) algorithm to estimate the parameters of the $Skew$-$t$ distribution, where we used the $R$ Package "\texttt{fitHeavyTail}", and the results are shown as:

\begin{gather*}
    \delta = 3.228143, \\[6pt]
    \mu =
    \begin{bmatrix}
        0.00321155 & -0.00042093 & 0.00231314 & 0.00124911
    \end{bmatrix}, \\[6pt]
    \gamma =
    \begin{bmatrix}
        -0.00039805 & 0.00011115 & -0.00005774 & 0.00000366
    \end{bmatrix},
\end{gather*}
and
\begin{equation*}
    \Sigma =
    \begin{bmatrix}
        0.00037775 & 0.00023791 & 0.00023987 & 0.00028738 \\
        0.00023791 & 0.00028480 & 0.00019535 & 0.00023228 \\
        0.00023987 & 0.00019535 & 0.00025751 & 0.00024117 \\
        0.00028738 & 0.00023228 & 0.00024117 & 0.00031692
    \end{bmatrix}.
\end{equation*}

We check that these data satisfies $\mu+\gamma EZ-\1 r_r\neq 0$. Based on these data we calculate the skewness-induced tangent portfolio (\ref{skew}) and the mean-std tangent portfolio (\ref{tangent}) (here $m=\mu+\gamma EZ$ and $\Omega=Var(Z)\gamma^T\gamma +EZ \Sigma$) as
\begin{equation}
\begin{split}
\mathcal{T}_{\mbox{skew}}=& 
\begin{bmatrix}
     1.05174647 & 2.08253672 & -1.33955922 & -0.79472397
\end{bmatrix},\\
\mathcal{T}_{mv}^{\small{nmvm}}= &
\begin{bmatrix}
     0.94092756 & 2.31151985 & -1.52934055 & -0.72310686 
\end{bmatrix}.\\
\end{split}
\end{equation}
It is interesting to observe here that the skewness-induced tangent portfolio $\mathcal{T}_{\mbox{skew}}$ is quite different from the mean-variance tangent portfolio $\mathcal{T}_{mv}^{\small{nmvm}}$. We recall here that $\mathcal{T}_{\mbox{skew}}$ is the optimal portfolio for any risk-averse expected utility maximizer and $\mathcal{T}_{mv}^{\small{nmvm}}$ is the optimal portfolio for the investors who care only on the mean and variance of their investment returns. The values of $\lambda_{U(a, b)}$ for different combinations of the parameters $(a, b)$ are listed in the following table when $W_0=5$. Figure 1 gives the graph of $\lambda_{U(a, b)}$ as a function of the parameters $(a, b)$.

\begin{table}[H]
\centering
\caption{Values of $\lambda_{U{(a, b)}}$ for different $(a, b)$ pairs.}
\begin{tabular}{c|ccccc}
\hline \hline
$W_0=5$ & $b=0.5$ & $b=1.0$ & $b=1.5$ & $b=2.0$ & $b=2.5$ \\ \hline 
$a = 1.5$ & 1.5241 & 1.5473 & 1.5841 & 1.6367 & 1.6999 \\
$a = 2.0$ & 1.1465 & 1.1621 & 1.1904 & 1.2295 & 1.2758 \\
$a = 2.5$ & 0.9172 & 0.9307 & 0.9520 & 0.9841 & 1.0203 \\
$a = 3.0$ & 0.7653 & 0.7759 & 0.7958 & 0.8196 & 0.8503 \\
$a = 3.5$ & 0.6569 & 0.6632 & 0.6815 & 0.7026 & 0.7284 \\ \hline \hline
\end{tabular}
\end{table}

\begin{figure}[H]
    \centering
    \includegraphics[width=0.9\linewidth]{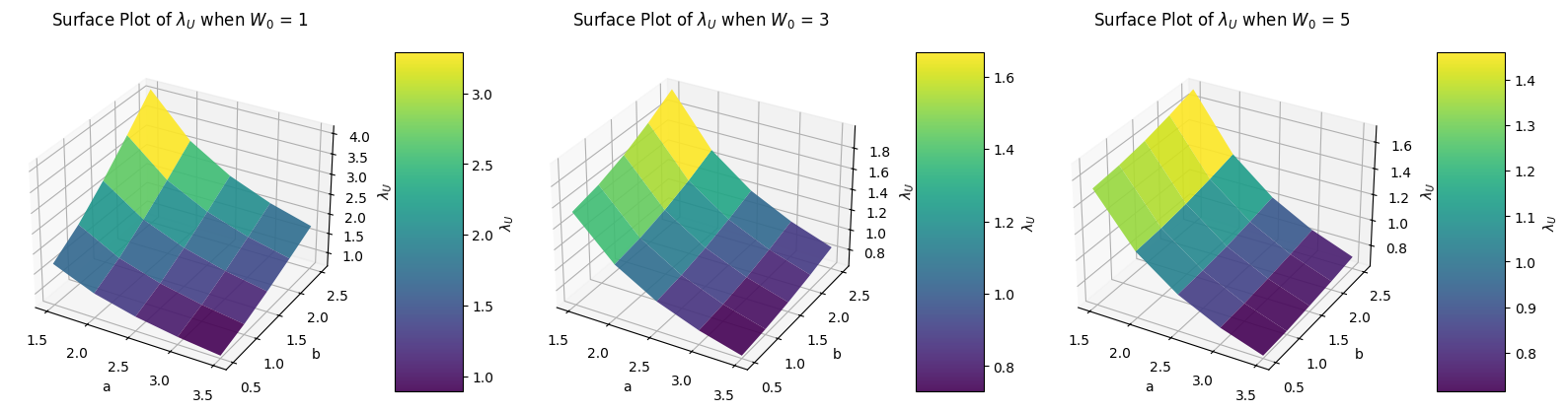}
    \caption{Surface in the three dimensional space $(a, b, \lambda_{U(a, b)})$ under the $Skew$-$t$ distribution}
    \label{fig_surface_stocks_parameters}
\end{figure}

Now we estimate the parameters of the log-normal model, as a testing to our example \ref{log_normal}, based on the same data for the above mentioned four stocks. Namely $X$ satisfies 
\[
X \overset{d}{=} \hat{\mu} + \hat{\Sigma}^{\frac{1}{2}} N_d,
\]
with
\[
\hat{\mu} =
\begin{bmatrix}
0.00229349 & 0.00214277 & 0.00098126 & 0.00133359
\end{bmatrix},
\]
and
\[
\hat{\Sigma} =
\begin{bmatrix}
0.00086637 & 0.00050060 & 0.00054323 & 0.00068924 \\
0.00050060 & 0.00058552 & 0.00039930 & 0.00049770 \\
0.00054323 & 0.00039930 & 0.00059199 & 0.00058189 \\
0.00068924 & 0.00049770 & 0.00058189 & 0.00076986
\end{bmatrix}.
\]

Based on these data, we calculate the tangent portfolio as
\[
\mathcal{T}_{mv}^{normal}=
\begin{bmatrix}
1.12344127 & 1.30546066 & -0.43472208 & -0.99417985
\end{bmatrix}.
\]
Clearly in this case the skewness-induced tangent portfolio $\hat{\mathcal{T}}_{\mathrm{skew}}$ is the same with $\mathcal{T}_{mv}^{normal}$. The values of $\hat{\lambda}_{U(a,b)}$ for different combinations of parameters $(a,b)$ are listed in the following table.
\begin{table}[H]
\centering
\caption{Values of $\hat{\lambda}_{U{(a, b)}}$ with $W_0 = 5$: log-normal case}
\begin{tabular}{c|ccccc}
\hline \hline
$W_0=5$ & $b=0.5$ & $b=1.0$ & $b=1.5$ & $b=2.0$ & $b=2.5$ \\ \hline 
$a = 1.5$ & 2.1245 & 2.1403 & 2.1691 & 2.2037 & 2.2513 \\
$a = 2.0$ & 1.5971 & 1.6101 & 1.6288 & 1.6540 & 1.6920 \\
$a = 2.5$ & 1.2794 & 1.2872 & 1.3016 & 1.3293 & 1.3547 \\
$a = 3.0$ & 1.0647 & 1.0743 & 1.0869 & 1.1063 & 1.1280 \\
$a = 3.5$ & 0.9141 & 0.9196 & 0.9328 & 0.9452 & 0.9671 \\\hline \hline
\end{tabular}
\end{table}
\begin{figure}[H]
    \centering
    \includegraphics[width=0.9\linewidth]{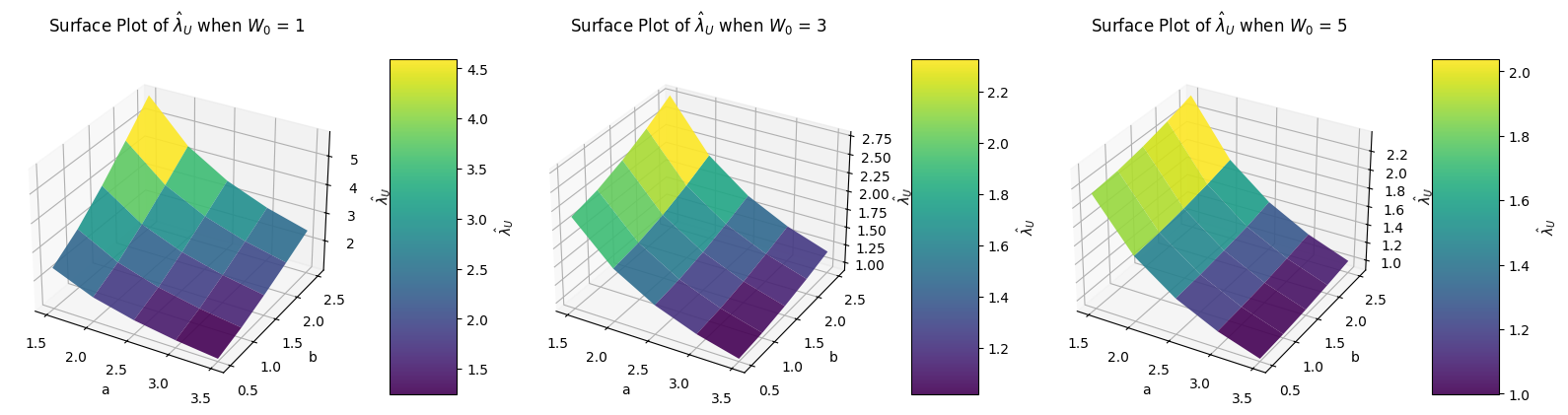}
    \caption{Surface in the three dimensional space $(a, b, \lambda_{U(a, b)})$ under the log-normal distribution}
    \label{fig_surface_stocks_parameters}
\end{figure}

Next we conduct testing on our Theorem  \ref{3.155}. To this end, we set the distribution parameters as mentioned above. The annual risk-free rate is set at $1.25\%$, with an initial wealth of $W_0 = 10$. The parameters for the SAHARA utility function are $a=1.5$, $b=5$. To assess the sensitivity of our theorem, we conducted a uniformly sampling of 100,000 instances from the interval \([-2 \times x_i, 4 \times x_i]\) for values of \(x_i \geq 0\), and from the interval \([4 \times x_i, -2 \times x_i]\) for values of \(x_i < 0\). Subsequently, we calculated the expected utility for each sample and identified the maximum value within that parameter setting.

\begin{figure}[H]
    \centering
    \includegraphics[width=1\linewidth]{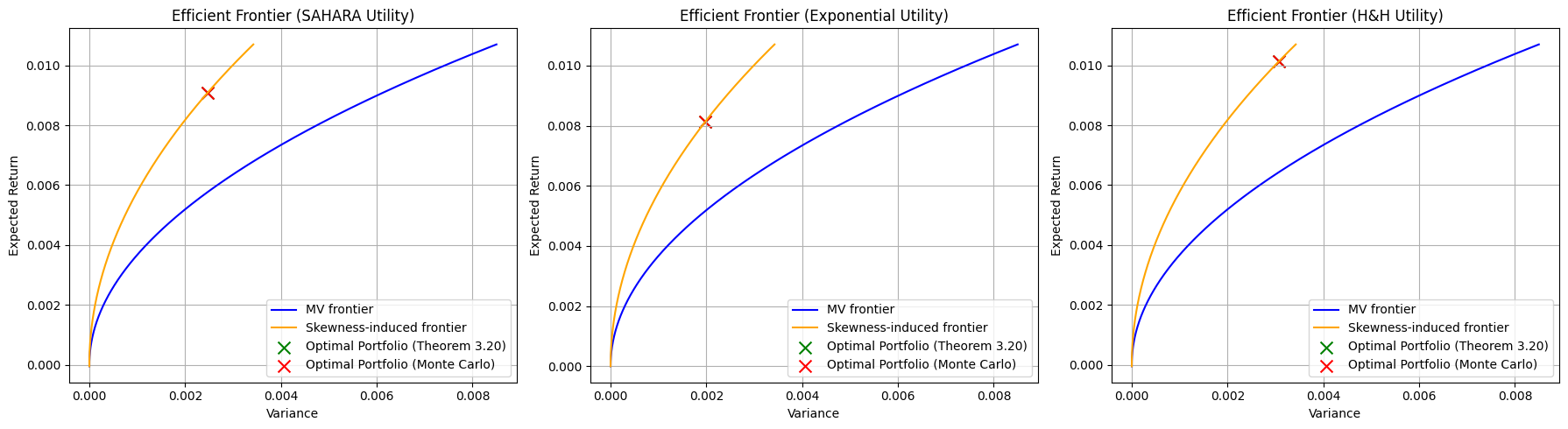}
    \caption{MV frontier (blue), Skewness-induced frontier (yellow), Optimal portfolio}
    \label{EF_Utility}
\end{figure}

The above Figure 3 presents the standard mean-variance frontier curve (cobalt blue curve) and the skewness-induced frontier curve (yellow curve). The mean and variance of the optimal portfolio (using $\bar{Y}$ defined above as return vector) in Theorem \ref{3.122} are found by using both (\ref{main-xstar}) and the Monte-Carlo methods. The corresponding points clearly lie on the skewness-induced frontier and they are extremely close as can be seen from the graphs. In the following Table 4 we test the accuracy of  the optimal portfolio (\ref{main-xstar}) by comparing it to the optimal portfolio obtained by using Monte-Carlo methods.  As it can be seen easily, these portfolios are extremely close to each other also.

\begin{table}[H]
\caption{Optimal portfolios under different utility functions}
\label{Table_All}
\centering
\begin{tabular}{c|cc|cc|cc}
\hline\hline
\multirow{2}{*}{Utility} 
 & \multicolumn{2}{c|}{\textbf{SAHARA}} 
 & \multicolumn{2}{c|}{\textbf{Exponential}} 
 & \multicolumn{2}{c}{\textbf{Henderson--Hubson}} \\
 & \multicolumn{2}{c|}{$a=1.5,\, b=5$} 
 & \multicolumn{2}{c|}{$\alpha=0.15$} 
 & \multicolumn{2}{c}{$\tau=0.1$} \\ \hline
Method 
 & Theorem \ref{3.155} & MC 
 & Theorem \ref{3.155} & MC 
 & Theorem \ref{3.155} & MC \\ \hline
Stock 1 & 1.787827 & 1.787767 & 1.597969 & 1.597988 & 1.993506 & 1.993621 \\
Stock 2 & 3.540032 & 3.539912 & 3.164098 & 3.164136 & 3.947291 & 3.947519 \\
Stock 3 & -2.277070 & -2.277147 & -2.035257 & -2.035233 & -2.539033 & -2.538887 \\
Stock 4 & -1.350924 & -1.350969 & -1.207462 & -1.207448 & -1.506339 & -1.506252 \\ \hline
ED & \multicolumn{2}{c|}{0.000161} & \multicolumn{2}{c|}{0.000050} & \multicolumn{2}{c}{0.000306} \\
CE & \multicolumn{2}{c|}{10.128456} & \multicolumn{2}{c|}{10.059688} & \multicolumn{2}{c}{9.950811} \\ \hline\hline
\end{tabular}
\vspace{1ex}
\begin{minipage}{0.9\linewidth}
\footnotesize
\textit{Notes:} ED denotes the Euclidean Distance between the portfolio weights obtained from the Theorem \ref{3.155} and the Monte Carlo (MC) method.  
CE denotes the Certainty Equivalent, which is the certain wealth level yielding the same expected utility as the optimal portfolio.
\end{minipage}
\end{table}

\section{Appendix: Proof of well posedness} 
Here we show that when the utility  function  $U$ satisfies  Assumption 1 and the model (\ref{one}) satisfies Assumption 2, the problem (\ref{new-L2}) is well-posed. To make our discussions convenient, we 
transform the portfolio space $R^d$ by a transformation that will be introduced below.
First note  that , with (\ref{one}), we have $X-\1 r_f=(\mu-\1 r_f)+\gamma Z+\sqrt{Z}AN_d$. We introduce a linear one-to-one transformation $\mathcal{T}:\R^d\rightarrow \R^d$, that maps $x\in \R^d$ into $y\in \R^d$ as  $y^T=x^TA$. We denote by $A_1^c, A_2^c, \cdots, A_d^c$ the column vectors of $A$ and express both $\mu-\1r_f$ and $\gamma$ as linear combinations of  $A_1^c, A_2^c, \cdots, A_d^c$, i.e., $\mu-\1 r_f=\sum_{i=1}^d\mu_i^{0}A_i^c,\; \gamma=\sum_{i=1}^d\gamma_i^{0}A_i^c$.
We denote by $\mu_0$ and $\gamma_0$ the column vectors of the coefficients of the above linear transformation, i.e., $\mu_0=(\mu_1^0, \mu_2^0, \cdots, \mu_d^0)^T, \;  \gamma_0=(\gamma_1^0, \gamma_2^0, \cdots, \gamma_d^0)^T$. Then for any portfolio $x$ we have $x^T(X-\1 r_f)\overset{d}{=}y^T\mu_0+y^T\gamma_0Z+|y|\sqrt{Z}N(0, 1)$, where $y^T=x^TA$ and $|\cdot|$ denotes the Euclidean norm of vectors. We have $y^T\mu_0=|y||\mu_0|Cos(y, \mu_0)$ and $y^T\gamma_0=|y||\gamma_0|Cos(y, \gamma_0)$, where $Cos(y, \mu_0)=(\mu_0 \cdot y)/ |\mu_0||y|$ and $Cos(y, \gamma_0)=(\gamma_0 \cdot y)/|\gamma_0||y|$ denote  the cosines of the angles between the vectors $y$ and $\mu_0$ and $y$ and $\gamma_0$ respectively. From now on we denote $\phi_y=Cos[(\gamma_0, y)]\;\;  \mbox{and} \;\;  \psi_y=Cos[(\mu_0, y)]$ for notational convenience. Observe that
$W(x)\overset{d}{=} W_0(1+r_f)+W_0[y^T\mu_0+y^T\gamma_0Z+|y|\sqrt{Z}N(0, 1)]$, whenever $x$ and $y$ are related by $y^T=x^TA$. For convenience, we also introduce the following notation $W(y)=:W_0(1+r_f)+W_0|y|\Big [|\mu_0|\psi_y+|\gamma_0|\phi_y Z+\sqrt{Z}N(0, 1)\Big ]$ and with this we have $W(x)\overset{d}{=}W(y)$ as long as $y^T=x^TA$. Clearly finding the solutions of (\ref{new-L2}) is equal to finding the solutions of 
$\max_{y\in \R^d}\; EU(W(y))$. 

Next we discuss the existence of a solution for (\ref{new-L2}). We assume that, as a minimal requirement, the utility function $U:\R\rightarrow \R$ is finite valued and non-decreasing. First we need to introduce some condition on the pair $(U, X)$ so that $EU(W(x))<+\infty$, whenever the Euclidean norm $|x|$ of the portfolio $x$ is finite, i.e. $|x|<+\infty$. For a similar discussion see Proposition 1 of \cite{Pirvu_Kwak}, where the expected utility under the cumulative prospect theory utility function was shown to be finite for any portfolio with a finite Euclidean norm when the return vector follows a skewed student  $t$-distribution. To this end, for any $\delta=(\delta_1, \delta_2, \delta_3)$ with $\delta_i\geq 0, 1\le i\le 3,$ we define $
X_{\delta}=\delta_1+\delta_2Z+\delta_3\sqrt{Z}|N(0, 1)|$, where $Z$ is the mixing distribution in (\ref{one}) and $N(0, 1)$ is any standard normal random variable independent from $Z$. Clearly, 
$X_{\delta}, \delta\in \R_3^{+},$ are non-negative random variables. We first write down the following simple Lemma.

\begin{lemma}\label{lem3.55} Consider an economy $(U, X)$ with $U:\R\rightarrow \R$ finite valued and non-decreasing. If 
\begin{equation}\label{the-con}
EU(X_{\delta})<+\infty    
\end{equation}
for all $\delta \in \R_3^{+}$, then $EU(W(x))<+\infty$ for any portfolio $x\in \R^d$ with $|x|<+\infty$. Hence if the condition (\ref{the-con}) holds and at the same time   
\begin{equation}\label{58888}
\sup_{x\in \R^n}EU(W(x))=+\infty,
\end{equation}   
happens then the economy $(U, X)$ admits AOP, i.e., there exists a sequence of portfolios $x_n$ with $|x_n|\rightarrow +\infty$ such that $EU(W(x_n))\rightarrow +\infty$ while there is no any portfolio with finite Euclidean norm that gives $+\infty$ expected utility.
\end{lemma}
\begin{proof} First note that $EU(W(x))=EU(W(y))$. We have $|y|=x^T\Sigma x$. Since $\Sigma$ is positive definite, there exists a constant $K>0$ such that $|y|=x^T\Sigma x\le K|x|$ for all $x\in \R^d$. For any positive number $m\geq 0$
define $D_m=\{y\in \R^d: |y|\le m \},\; \tilde{W}_m=W_0(1+r_f)+W_0m \left [|\mu_0|+|\gamma_0|Z+\sqrt{Z}|N|\right ]$. Also for any $y\in \R^d$  define $\tilde{W}_y=W_0(1+r_f)+W_0|y| [|\mu_0|+|\gamma_0|Z+\sqrt{Z}|N|]$. Observe that $W(y)\le \tilde{W}_y$ almost surely. Since the utility function $U$ is non-decreasing we have $EU(W(y))\le EU(\tilde{W}(y))$. On the domain $D_m$ we have $\tilde{W}_y\le \tilde{W}_m$ for all $m\geq 0$. Again since $U$ is non-decreasing we have $EU(\tilde{W}_y)\le EU(\tilde{W}_m)$ on the domain $D_m$ for each $m\geq 0$. The stated condition in the Lemma implies that $EU(\tilde{W}_m)<+\infty$ for each $m\geq 0$. From these we conclude that $EU(W(y))<+\infty$ on $D_m$ for each $m\geq 0$. Now for any portfolio $x\in \R^d$ with finite $m_0=:|x|$ we have $|y|\le m_0K$. Thus $y\in D_{m}$ with $m=Km_0$. This implies that $EU(W(x))=EU(W(y))<+\infty$. 
 
 To see the second part of the claim in the Lemma, note that for any portfolio $x$ with finite Euclidean norm we always have $EU(W(x))<+\infty$ under the condition (\ref{the-con}) as shown above. Therefore it is sufficient to rule out the possibility of the existence of a sequence $\{x_n\}$ of portfolios with uniformly bounded Euclidean norm  such that $EU(W(x_n))\rightarrow +\infty$. By the way of contrary assume that there is such sequence 
 $\{x_n\}$. Then it has a convergent sub-sequence 
 $x_{n_k}$ with $EU(W(x_{n_k}))\rightarrow +\infty$
 (such a sub-sequence exists as $\{|x_n|\}$ is bounded). Then all of  
 $\{x_{n_k}^T\mu\}, \{x_{n_k}^T\gamma\}, \{x_{n_k}^T\Sigma x_{n_k}\}$ are bounded families. Hence their exists $\delta_1>0, \delta_2>0, \delta_3>0$ such that $W(x_{n_k})\le X_{\delta}$ almost surely for all positive integer $k$. The utility function is non-decreasing hence we have 
 $EU(W(x_{n_k}))\le EU(X_{\delta})<+\infty$, a contradiction.
 \end{proof}

\begin{lemma}\label{prop3.6} Consider an economy $(U, X)$ with $U: \R\rightarrow \R$ finite-valued and non-decreasing. Assume $U$ is concave and the mixing distribution $Z$ in (\ref{one}) has finite first moment. Then for any portfolio $x$ with $|x|<+\infty$ we have $EU(W(x))<+\infty$. Hence under the condition that $U$ is concave and $Z\in L^1$ in the model (\ref{one}), if
\[
\sup_{x\in \R^n}EU(W(x))=+\infty,
\]
then the economy $(U, X)$ admits AOP.
\end{lemma}
\begin{proof} By Lemma \ref{lem3.55} we need to show (\ref{the-con}). Since $U$ is concave we have
\[
EU(X_{\delta})\le U(EX_{\delta})=U(\delta_1+\delta_2EZ+\delta_3E\sqrt{Z}E|N|).
\]
Since $Z\in L^1$ we have $\sqrt{Z}\in L^2\subset L^1$. Therefore $\delta_1+\delta_2EZ+\delta_3E\sqrt{Z}E|N|$ is a finite number for each fixed $\delta$. Then since $U$ is finite valued we clearly have $U(\delta_1+\delta_2EZ+\delta_3E\sqrt{Z}E|N|)<+\infty$. 
Then from Lemma \ref{lem3.55} we know that for any portfolio $x$ with finite Euclidean norm  we have 
$EU(W(x))<+\infty$. Thus the remaining claims in the Lemma holds.     
\end{proof}
Clearly AOP are costly since  with AOP one has to invest infinite amount on the risky assets. Hence an economy $(U, X)$ that admits AOP is impractical for an expected utility maximizer as financial resources are always limited. This illustrates the need for introducing  some sufficient conditions on $(U, X)$ that are necessary for the exclusion of AOP. Below we discuss this problem.

In the rest of this section, we discuss some conditions on the utility function $U$ that can rule out the possibility (\ref{58888}). Note first that for the trivial portfolio $x=0\in \R^d$ (investing everything on the risk-free asset) we have $EU(W(0))=U(W_0(1+r_f))>-\infty$. Therefore we always have $\sup_{x\in \R^d}EU(W(x))>-\infty$ as long as $U$ is a finite valued utility function. Recall that with the transformation $y=x^TA$ introduced in the introduction we have $EU(W(x))=EU(W(y))$ as long as $y=x^TA$. Therefore in the rest of this section we work at the ``$y-$coordinate system'' and study some conditions on $U$ that can rule out (\ref{58888}). To this end, we define the following sets first $D_{\pi}=:D(\pi_1, \pi_2)=:\{y\in \R^n: \phi_y=\pi_1, \psi_y=\pi_2\}$  for any $\pi=(\pi_1, \pi_2)\in I=:[-1, 1]\times [-1, 1]$. For convenience we introduce the following notations:
$\xi_{\pi}=:|\mu_0|\pi_1+|\gamma_0|\pi_2 Z+\sqrt{Z}N(0, 1), \; \; \pi  \in I$. Observe that for each fixed $\pi$, the random variable $\xi_{\pi}$ has support in $(-\infty, +\infty)$ as long as $Z$ is not trivial, i.e., $P(Z>0)>0$. With this notation we have $W(y)=W_0(1+r_f)+W_0|y|\xi_{\pi}$,  when $y\in D_{\pi}$. From this it is easy to see that for each fixed $\pi\in I$ we have $\lim_{|y| \rightarrow 0^+}W(y)\overset{a.s.}{\rightarrow} \xi_0,   \; \; \lim_{|y| \rightarrow +\infty}
W(y)\overset{a.s.}{\rightarrow} \xi_{\infty}^{\pi}$ when $y\in D_{\pi}$, 
where $\xi_0=W_0(1+r_f)$ is a constant and $\xi_{\infty}^{\pi}$ equals to $+\infty$ if $\xi_{\pi}> 0$ and equals to $-\infty$ if $\xi_{\pi}\le 0$. Since $\xi_{\pi}$ has full support as long as $P(Z>0)> 0$, we have  $P(\xi_{\pi}>0)>0$ and  $P(\xi_{\pi}\le 0)>0$ for each fixed $\pi \in I$. Hence the limit random variable $\xi_{\infty}^{\pi}$ above is non-trivial.
\begin{lemma} \label{lem3.4} Consider the optimization problem 
\begin{equation}\label{wyyy}
\max_{y\in \R^d}\; EU(W(y)).
\end{equation}
\begin{enumerate}
\item [i)] If $U$ is finite valued,  bounded from below,  and $\lim_{w\rightarrow +\infty}U(w)=+\infty$, then 
 \begin{equation*}
\sup_{y\in \R^d}\; EU(W(y))=+\infty.
\end{equation*}  

\item [ii)] If $U$ is finite valued, continuous, bounded from above,  and $\lim_{w\rightarrow -\infty}U(w)=-\infty$, then (\ref{wyyy}) is well defined.
\end{enumerate}
\end{lemma}
\begin{proof} Assume $U$ is bounded from below and $\lim_{w\rightarrow +\infty}U(w)=+\infty$. For each fixed $\pi\in I$, when $y\in D_{\pi}$ we have
\begin{equation}\label{bd-u-d}
EU(W(y))=E[U(W_0(1+r_f)+W_0|y|\xi_{\pi})1_{\xi_{\pi}>0}]
+E[U(W_0(1+r_f)+W_0|y|\xi_{\pi})1_{\xi_{\pi}<0}].
\end{equation}
On the event $\{\xi_{\pi}>0\}$, when $|y|\rightarrow +\infty$ we have $W_0(1+r_f)+W_0|y|\xi_{\pi}\rightarrow +\infty$ almost surely. Therefore $U(W_0(1+r_f)+W_0|y|\xi_{\pi})\rightarrow +\infty$ almost surely also on $\{\xi_{\pi}>0\}$ as $\lim_{w\rightarrow +\infty}U(w)=+\infty$. Since $U$ is bounded from below, an application of Fatou's lemma gives
\[
\liminf_{|y|\rightarrow +\infty}E[U(W_0(1+r_f)+W_0|y|\xi_{\pi})1_{\xi_{\pi}>0}]\geq E\liminf_{|y|\rightarrow +\infty}[U(W_0(1+r_f)+W_0|y|\xi_{\pi})1_{\xi_{\pi}>0}]=+\infty.
\]
At the same time the second term on the right-hand-side of (\ref{bd-u-d}) is bounded from bellow. Hence we can conclude that 
$\liminf_{|y|\rightarrow +\infty}EU(W(y))\rightarrow +\infty$. 

Now assume $U$ is bounded from above by a real number $M$ and $\lim_{w\rightarrow -\infty}U(w)=-\infty$. We have $EU(W(y))\le M$ for all $y\in \R^n$. Denote $B=\sup_{y\in \R^d}EU(W(y))\le M$. Since $EU(W(0))=U(W_0(1+r_f))>-\infty$, we have $B>-\infty$. First we rule out the following case: there exists a sequence $y_n$ with 
$|y_n|\rightarrow +\infty$ such that
$\lim_{n\rightarrow +\infty}EU(W(y_n))=\sup_{y\in \R^d}EU(W(y))<+\infty$. Assume by contradiction that such  a sequence $\{y_n\}$ exists. Denote $\phi_{y_n}=Cos[(\gamma_0, y_n)]$ and $\psi_{y_n}=Cos[(\mu_0, y_n)]$.
Since $\phi_{y_n}$ and $\psi_{y_n}$ take values in $[-1, 1]$ we can assume that $y_n$ has a sub-sequence such that 
both $\phi_{y_n}$ and $\psi_{y_n}$ converges. Without loss of any generality we can assume $\phi_{y_n}\rightarrow \pi_1^0\in [-1, 1]$ and $\psi_{y_n}\rightarrow \pi_2^0\in [-1, 1]$. We define 
$\xi_{\pi^n}=|\mu_0|\phi_{y_n}+|\gamma_0|\psi_{y_n}Z+\sqrt{Z}N(0, 1)$ and $\xi_{\pi^0}=|\mu_0|\pi_1^0+|\gamma_0|\pi_2^0Z+\sqrt{Z}N(0, 1)$. Observe that $\xi_{\pi^n}\rightarrow \xi_{\pi^0}$ almost surely and $P(\xi_{\pi^0}>0)>0$. We have 
\begin{equation}\label{73-new}
\begin{split}
EU(W(y_n))=&E[U(W_0(1+r_f)+W_0|y_n|\xi_{\pi^n})1_{\xi_{\pi^n}>0}]\\
+&E[U(W_0(1+r_f)+W_0|y_n|\xi_{\pi_n})1_{\xi_{\pi^n}<0}].\\
\end{split}
\end{equation}
The first term on the right-hand-side of (\ref{73-new}) is bounded from above as $U$ is bounded from above. The second term on the right-hand-side of (\ref{73-new}) can be shown to converge to $-\infty$ when $|y_n|\rightarrow +\infty$ (by using Fatou's Lemma) as $U$ is continuous and $\lim_{w\rightarrow -\infty}U(w)=-\infty$ by the assumption on $U$. Hence we can conclude that $EU(W(y_n))\rightarrow -\infty$ as $|y_n|\rightarrow +\infty$. But $sup_{y\in \R^d}EU(W(y))\geq EU(W(0))=U(W_0(1+r_f))>-\infty$. Thus $\lim_{n\rightarrow +\infty}EU(W(y_n))=\sup_{y\in \R^d}EU(W(y))<+\infty$ can not happen.

Now, by the definition of $B$ there exists a sequence $\bar{y}_n$ such that $\lim_{n\rightarrow +\infty}EU(W(\bar{y}_n))=B$. By the above analysis the sequence $\bar{y}_n$ can't have a sub-sequence, which we denote by itself $\bar{y}_n$ for the sake of notational simplicity, such that $|\bar{y}_n|\rightarrow +\infty$. Therefore $\{|\bar{y}_n|\}$ is a bounded sequence. Hence we can conclude that there exists a vector $\bar{y}_0\in \R^d$ and a sub-sequence of $\{\bar{y}_n\}$, which we denote by itself again, such that $\bar{y}_n\rightarrow \bar{y}_0$. Then $W(\bar{y}_n)\rightarrow W(\bar{y}_0)$ almost surely and since $U$ is continuous we have $U(W(\bar{y}_n))\rightarrow U(W(\bar{y}_0))$ almost surely. Since $U$ is bounded from above, the family $\{-U(W(\bar{y}_n))\}$ is bounded from below. By Fatou's lemma we have
\[
\liminf_nE[-U(W(\bar{y}_n))]\geq E\liminf_n[-U(W(\bar{y}_n))]=-EU(W(\bar{y}_0)).
\]
From this we conclude that $\limsup_nE[U(W(\bar{y}_n))]\le EU(W(\bar{y}_0))$, which implies $EU(W(\bar{y}_0))=\sup_{y\in \R^d}EU(W(y))$.
\end{proof}

We remark here that the above Lemma \ref{lem3.4} gives some sufficient conditions on the utility function $U$ for the well-posedness of the problem (\ref{new-L2}). Clearly the exponential utility functions $U(w)=-e^{aw}, a>0,$ satisfy the  conditions stated in the second half of this Lemma. Hence the problem (\ref{new-L2}) with exponential utility is well-posed, see the recent paper \cite{Rasonyi-Sayit} for this. Other than these, the utility functions that are presented in Example \ref{short} below and the class of Sahara utility functions with certain parameters (see Lemma \ref{3.8} below) also satisfy the sufficiency for the well-posedness of the problem (\ref{new-L2}). If the utility function $U$ satisfies 
\begin{equation}\label{both-inf}
 \lim_{x\rightarrow +\infty}U(x)=+\infty, \; \; \lim_{x\rightarrow -\infty}U(x)=-\infty,   
\end{equation}
then the well-posedness or the existence of the AOP in the economy $(U, X)$ depends on the properties of $U$ and also on $X$. The paper  \cite{Pirvu_Kwak} in its Proposition 2 shows the well-posedness of the expected utility maximization problem under the cumulative prospect theory utility function 
when $X$ is a skewed student  $t$-distribution. The cumulative prospect theory utility function clearly satisfies (\ref{both-inf}). Obtaining a sufficient condition on the utility function $U$ with (\ref{both-inf}) that guarantee the well-posedness of the problem (\ref{new-L2}) for any given model (\ref{one}) seems difficult. Below we present an example that demonstrates that when the utility function $U$ satisfies (\ref{both-inf}), the problem (\ref{new-L2}) is well-posed for some models (\ref{one}) and the economy $(U, X)$ admits AOP in some cases.

\begin{example} Consider the model (\ref{one}) in the Example \ref{ex-fi} with the corresponding wealth $W(x)=W_0+xW_0N(0, 1)$. Take the following utility function
\begin{equation*}
 U(x)=\left \{
\begin{array}{cc}
 k_1x    &  x\geq 0, \\
 k_2x    & x<0.
\end{array}
 \right.
\end{equation*}
for some $k_1>0, k_2>0$. For $x>0$ we can easily calculate
\[
EU(W(x))=W_0k_1W_0+W_0(k_2-k_1)\Phi(-\frac{1}{x})+\frac{W_0x}{\sqrt{2\pi}}e^{-\frac{1}{2x^2}}(k_1-k_2).
\]
The term  $W_0k_1W_0+W_0(k_2-k_1)\Phi(-\frac{1}{x})$ in this expression is a bounded number for all $x>0$. Therefore if $k_1<k_2$ then when $x\rightarrow +\infty$ we have $EU(W(x))=-\infty$. Since $W(x)\overset{d}{=}W(-x)$, we can hence conclude that when $k_1<k_2$, the problem (\ref{new-L2}) is well-posed. On the other hand if $k_2>k_1$, then when $x\rightarrow +\infty$ we have $EU(W(x))=+\infty$. In this case the economy $(U, X)$ admits AOP.
\end{example}
If the utility function $U$ satisfies (\ref{both-inf}), then the  limit $\lim_{|y| \rightarrow +\infty}U(W(y))$ when $y\in D_{\pi}$
may do not exist. Even if it exists and 
$EU(\xi_{\infty}^{\pi})$ is well defined it may happen that 
$\max_{x\in \R^d}EU(W(x))\le EU(\xi_{\infty}^{\pi})$ for some $\pi \in I$. Hence for utility functions $U$ with (\ref{both-inf}), the well-posedness of the problem (\ref{new-L2}) needs some care. A similar problem were studied in Proposition 2 of \cite{Pirvu_Kwak} when the utility function $S-$ shaped utility function and when the return vector has a skewed student $t$-distribution. The conditions in Assumption 1 above does not guarantee that the map $y\rightarrow EU(W(y))$ is continuous. As discussed in Example 1 in \cite{Rasonyi-Sayit}, when $U(w)=e^{-aw}, a>0,$ is exponential utility  and when $\gamma=0, Z=e^{N(0, 1)}$ in the model (\ref{one}) one has $EU(W(0))=U(W_0(1+r_f))$ and $EU(W(y))=-\infty$ for any other $y\neq 0$. Clearly in this case the map $y\rightarrow EU(W(y))$ is not continuous. Hence it is not immediately clear if the optimization problem (\ref{new-L2}) always has a solution when the domain $D$ is a bounded and closed subset of $\R^d$ under Assumption 1. 

In the following Lemma we show that the problem (\ref{L2}) is well-posed for any closed domain $D$ (no need to assume $D$ is bounded) and for any model (\ref{one}) as long as the utility function $U$ satisfies  Assumption 1 above.

\begin{lemma}\label{opt-D} Assume the utility function $U$ satisfies  Assumption 1. Let $D$ be any closed subset of $\R^d$  with a vector $x_0\in D$ such that $EU(W(x_0))>-\infty$. Then the map $x\rightarrow EU(W(x))$ is upper semi-continuous and 
the problem (\ref{L2}) always  has a solution for any given model (\ref{one})), i.e., there exists a $x_0\in D$ with $EU(W(x))\le EU(W(x_0))$ for all $x\in D$ and $EU(W(x_0))>-\infty$.    
\end{lemma}
\begin{proof} We have $EU(W(x))<+\infty$ for all $x\in D$ as $U$ is bounded above. Also $EU(W(0))=U(W_0(1+r_f))>-\infty$. Define the map $e: x\rightarrow EU(W(x))$ and let $e_0=\sup_{x\in D}e(x)$. Then $-\infty<e_0<+\infty$. By the definition of $e_0$, we have a sequence $x_n\in D$ such that $e(x_n)=EU(W(x_n))\rightarrow e_0$. Without loss of any generality we can assume that $-\infty<e(x_n)<+\infty$ for all $x_n$. We claim that the family $\{x_n\}$ is bounded in the Euclidean norm. If not then there exists a sub-sequence $x_{n_k}$ with $|x_{n_k}|\rightarrow +\infty$. Then by ii) of Lemma \ref{lem3.4} we have
$e(x_k)=EU(W(x_{n_k}))\rightarrow -\infty$. A contradiction. Hence $\{x_n\}$ is a bounded family.
Therefore the sequence $x_n$ has a convergent sub-sequence to a limit $x_0\in D$ (as $D$ is a closed subset). 
Without loss of any generality we assume $x_n\rightarrow x_0$. Then $W(x_n)$ converges to $W(x_0)$ almost surely and since $U$ is continuous we have $U(W(x_n))\rightarrow U(W(x_0))$ almost surely. Since $\{U(W(x_n))\}$ is uniformly bounded from above, the family \{$-U(W(x_n))$\} is a sequence of random variables bounded from below. Then by Fatuo's lemma we have
\begin{equation*}\label{68}
\begin{split}
-EU(W(x_0))=E\liminf_n[-U(W(x_n))]&\le \liminf_nE[-U(W(x_n)]\\
&=-\lim sup_nEU(W(x_n))=-e_0.
\end{split}
\end{equation*}
This shows that $e_0\le EU(W(x_0))$. Then the optimality of $e_0$ implies that $e_0=EU(W(x_0))$. Last, it is easy to see from the relation (\ref{68}) that the map $x\rightarrow EU(W(x))$ is upper semi-continuous. This completes the proof.
\end{proof}

\begin{remark} The above Proposition (\ref{opt-D}) shows in particularly that the optimization problem (\ref{new-L2}) with $D=\mathcal{S}$ (the portfolio domain with short-sales  constraints defined in (\ref{short-33})) always has a solution and the optimizing portfolio also belongs to $\mathcal{S}$ under Assumption 1 on the utility function $U$. 
\end{remark}

We summarize the analysis in this section in the following Theorem. 

\begin{theorem}\label{cor-well} Consider an economy 
$(U, X)$. Assume $U$ satisfies Assumption 1. Then the problem (\ref{new-L2}) is well-posed for any given model
(\ref{one}) with the mixing distribution $Z$ can be any non-negative finite valued random variable.
\end{theorem}
\begin{proof} It is clear that under Assumption 1, the set $A=:\{x\in \R^d: U(W(x))\geq U(W(0))\}$ is a closed and bounded subset of $\R^d$. Also 
observe that $\max_{x\in \R^d}\; EU(W(x))=\max_{x\in A}\; EU(W(x))$. Then the claim in the Theorem follows from Lemma \ref{opt-D}.
\end{proof}
\begin{remark} The condition $\lim_{w\rightarrow -\infty}U(w)=-\infty$ on the utility function in Assumption 1 above is important for our Theorem \ref{cor-well} above. This is partly due to the fact that our Theorem \ref{cor-well} is stated for any model (\ref{one}). For a particular given model (\ref{one}), the well-posedness of the problem (\ref{one}) may hold under much weaker conditions on the utility function $U$ than the conditions in Assumption 1 above.  A similar condition were discussed in Proposition 2 of \cite{Pirvu_Kwak}. As stated in the paragraph proceeding to Proposition 2 in this paper, the condition  $\lim_{w\rightarrow -\infty}U(w)=-\infty$ guarantees that any unlimited investment in the risky assets is worse than zero investment in the risky assets.
    
\end{remark}

\end{document}